\documentclass[aps,prfluids,preprint,superscriptaddress,amsmath,amssymb]{revtex4-2}

\usepackage{graphicx}
\usepackage{amsthm}
\usepackage[hidelinks]{hyperref}
\usepackage[capitalise,nameinlink]{cleveref}
\usepackage{booktabs}
\usepackage{subcaption}
\usepackage{multirow}
\usepackage{adjustbox}
\usepackage{url}
\usepackage{placeins}
\usepackage{xcolor}

\graphicspath{{./}{figs/}}
\DeclareMathOperator{\diag}{diag}
\theoremstyle{plain}
\newtheorem{lemma}{Lemma}

\newtheorem{corollary}{Corollary}
\theoremstyle{remark}
\newtheorem{remark}{Remark}
\RequirePackage[normalem]{ulem} 
\RequirePackage{xcolor} 
\providecommand{\DIFaddbegin}{} 
\providecommand{\DIFaddend}{} 
\providecommand{\DIFdelbegin}{} 
\providecommand{\DIFdelend}{} 
\providecommand{\DIFaddbeginFL}{} 
\providecommand{\DIFaddendFL}{} 
\providecommand{\DIFdelbeginFL}{} 
\providecommand{\DIFdelendFL}{} 
\newcommand{\blue}[1]{\texorpdfstring{{\protect\color{blue}#1}}{#1}}

\newcommand{\DIFscaledelfig}{0.5}
\RequirePackage{settobox} 
\RequirePackage{letltxmacro} 
\newsavebox{\DIFdelgraphicsbox} 
\newlength{\DIFdelgraphicswidth} 
\newlength{\DIFdelgraphicsheight} 
\LetLtxMacro{\DIFOincludegraphics}{\includegraphics} 
\newcommand{\DIFaddincludegraphics}[2][]{{\color{blue}\fbox{\DIFOincludegraphics[#1]{#2}}}} 
\newcommand{\DIFdelincludegraphics}[2][]{
\sbox{\DIFdelgraphicsbox}{\DIFOincludegraphics[#1]{#2}}
\settoboxwidth{\DIFdelgraphicswidth}{\DIFdelgraphicsbox} 
\settoboxtotalheight{\DIFdelgraphicsheight}{\DIFdelgraphicsbox} 
\scalebox{\DIFscaledelfig}{
\parbox[b]{\DIFdelgraphicswidth}{\usebox{\DIFdelgraphicsbox}\\[-\baselineskip] \rule{\DIFdelgraphicswidth}{0em}}\llap{\resizebox{\DIFdelgraphicswidth}{\DIFdelgraphicsheight}{
\setlength{\unitlength}{\DIFdelgraphicswidth}
\begin{picture}(1,1)
\thicklines\linethickness{2pt} 
{\color[rgb]{1,0,0}\put(0,0){\framebox(1,1){}}}
{\color[rgb]{1,0,0}\put(0,0){\line( 1,1){1}}}
{\color[rgb]{1,0,0}\put(0,1){\line(1,-1){1}}}
\end{picture}
}\hspace*{3pt}}} 
} 
\LetLtxMacro{\DIFOaddbegin}{\DIFaddbegin} 
\LetLtxMacro{\DIFOaddend}{\DIFaddend} 
\LetLtxMacro{\DIFOdelbegin}{\DIFdelbegin} 
\LetLtxMacro{\DIFOdelend}{\DIFdelend} 
\DeclareRobustCommand{\DIFaddbegin}{\DIFOaddbegin \let\includegraphics\DIFaddincludegraphics} 
\DeclareRobustCommand{\DIFaddend}{\DIFOaddend \let\includegraphics\DIFOincludegraphics} 
\DeclareRobustCommand{\DIFdelbegin}{\DIFOdelbegin \let\includegraphics\DIFdelincludegraphics} 
\DeclareRobustCommand{\DIFdelend}{\DIFOaddend \let\includegraphics\DIFOincludegraphics} 
\LetLtxMacro{\DIFOaddbeginFL}{\DIFaddbeginFL} 
\LetLtxMacro{\DIFOaddendFL}{\DIFaddendFL} 
\LetLtxMacro{\DIFOdelbeginFL}{\DIFdelbeginFL} 
\LetLtxMacro{\DIFOdelendFL}{\DIFdelendFL} 
\DeclareRobustCommand{\DIFaddbeginFL}{\DIFOaddbeginFL \let\includegraphics\DIFaddincludegraphics} 
\DeclareRobustCommand{\DIFaddendFL}{\DIFOaddendFL \let\includegraphics\DIFOincludegraphics} 
\DeclareRobustCommand{\DIFdelbeginFL}{\DIFOdelbeginFL \let\includegraphics\DIFdelincludegraphics} 
\DeclareRobustCommand{\DIFdelendFL}{\DIFOaddendFL \let\includegraphics\DIFOincludegraphics} 
\makeatletter 
\let\sout@orig\sout 
\renewcommand{\sout}[1]{\ifmmode\text{\sout@orig{\ensuremath{#1}}}\else\sout@orig{#1}\fi} 
\makeatother 
\RequirePackage{listings} 
\lstdefinelanguage{DIFcode}{ 
  moredelim=[il][\color{white}\tiny]{\%DIF\ <\ }, 
  moredelim=[il][\sffamily\bfseries]{\%DIF\ >\ } 
} 
\lstdefinestyle{DIFverbatimstyle}{ 
	language=DIFcode, 
	basicstyle=\ttfamily, 
	columns=fullflexible, 
	keepspaces=true 
} 
\lstnewenvironment{DIFverbatim}{\lstset{style=DIFverbatimstyle}}{} 
\lstnewenvironment{DIFverbatim*}{\lstset{style=DIFverbatimstyle,showspaces=true}}{} 
\begin{document}

\title{Deep Koopman Sensing}

\author{Nithin Somasekharan}
\email{somasn@rpi.edu}
\affiliation{Department of Mechanical, Aerospace, and Nuclear Engineering, Rensselaer Polytechnic Institute, 110 8th St, Troy, New York 12180, USA}

\author{Yadi Cao}
\email{yadi.cao@ucf.edu}
\affiliation{Department of Computer Science, Institute of Artificial Intelligence, University of Central Florida, 4356 Scorpius Street, Orlando, Florida 32816, USA}

\author{Shaowu Pan}
\email{pans2@rpi.edu}
\thanks{Corresponding author.}
\affiliation{Department of Mechanical, Aerospace, and Nuclear Engineering, Rensselaer Polytechnic Institute, 110 8th St, Troy, New York 12180, USA}
\affiliation{Scientific Computation Research Center, Rensselaer Polytechnic Institute, 110 8th St, Troy, New York 12180, USA}

\begin{abstract}
Real-time reconstruction of fluid flows from sparse sensor measurements is important for both physical understanding and flow control. When first-principles models are too expensive for online data assimilation (DA), learned reduced-order models provide an efficient alternative, but are commonly optimized for forward prediction rather than state estimation. We propose Deep Koopman Sensing, a data-driven reduced-order DA framework that combines a nonlinear autoencoder with parameter-conditioned linear latent dynamics approximating the Koopman operator. We compare the proposed model with parametric dynamic mode decomposition (pDMD), a multilayer perceptron (MLP), and xLSTM across four benchmarks: 1D viscous Burgers, 2D flow past a cylinder, 2D dambreak, and 3D flow past a sphere. Our results reveal a marked distinction between forecasting and sensing: open-loop accuracy does not reliably predict assimilation performance, while Deep Koopman Sensing achieves the lowest assimilation error across all four benchmarks. More importantly, with an extended Kalman filter, incorporating sensor measurements improves the estimates of both linear latent models across all four benchmarks, whereas it degrades the nonlinear models, despite their strong open-loop performance. With ensemble filtering, the nonlinear models are no longer degraded by assimilation, while the Koopman model still attains the lowest assimilation error. These results show that latent dynamics should be designed for the downstream estimation task rather than selected solely for forecast accuracy, and demonstrate Koopman-based reduced-order modeling as an effective approach for real-time flow reconstruction from sparse, streaming measurements.

\end{abstract}

\maketitle
\section{Introduction}\label{sec:intro}

Inferring unsteady flow fields from a small number of noisy measurements is a central problem in flow sensing and control. The state typically consists of a spatially discretized solution of a nonlinear partial differential equation, making both its evolution and its estimation high-dimensional problems. These high-dimensional nonlinear dynamics make real-time flow field inference difficult~\cite{brunton2020machine,lino2023current,brunton2015closed}. The challenge is not only to maintain accurate predictions but also to update them rapidly as new observations arrive.
Data assimilation (DA) addresses this by combining a known dynamical model with noisy observations to estimate the system state~\cite{carrassi2018data}. 
For linear-Gaussian systems, the Kalman filter is optimal, but direct application to fluid flows is usually infeasible because the state dimension is too large and the full dynamics are strongly nonlinear and computationally expensive~\cite{kalman1960new,evensen2009data,law2015data}. Reduced-order models therefore play a dual role: they lower the cost of forecasting and define the state space in which filtering is carried out \cite{benner2015survey}. In practice, the quality of the reduced dynamics matters as much for the DA update as it does for open-loop rollout.
Classical reduced-order modeling in fluids is based on linear subspaces, most notably proper orthogonal decomposition (POD) \cite{lumley1967structure,sirovich1987turbulence,holmes2012turbulence}. These methods are effective when the dominant structures are well represented by a low-dimensional linear basis, but transport-dominated and strongly transient flows are harder to compress in this way \cite{peherstorfer2022breaking}. This has motivated nonlinear latent representations learned by autoencoders and related neural architectures \cite{lee2020model,vlachas2018data,yeung2019learning}. Coupled with nonlinear predictors such as MLPs and recurrent networks, these models can achieve strong open-loop forecasting performance on complex dynamics \cite{hochreiter1997long,vlachas2018data}.

The challenge is that good rollout accuracy does not automatically imply good data-assimilation performance. Once a model is placed inside a data-assimilation framework (e.g., a Kalman filter), the filter interacts not only with the forecast itself but also with the local linearization of the forecast map and the resulting uncertainty propagation.
Koopman operator theory offers an elegant solution. It seeks a representation in which nonlinear dynamics evolve linearly in a lifted coordinate system \cite{koopman1931hamiltonian,mezic2005spectral,rowley2009spectral}. Dynamic Mode Decomposition (DMD) and its extensions provide data-driven approximations of this viewpoint \cite{schmid2010dynamic,tu2014dynamic,jovanovic2014sparsity,williams2015data}. More recent deep Koopman models combine nonlinear encoders and decoders with approximately linear latent evolution, thereby retaining expressive state representations while imposing structure on the dynamics \cite{lusch2018deep,takeishi2017learning,yeung2019learning}.
Koopman and DMD ideas have already shown promise for flow-field estimation from sparse measurements. DMD-Kalman approaches have been used for unsteady flow estimation around actuated airfoils, and Koopman-linear estimators have been proposed for wind-turbine wake reconstruction \cite{gomez2019data,chen2022dynamic}. Information-based sensor placement has likewise been coupled with data-driven linear estimators for unsteady flows \cite{graff2023information}. These studies suggest that linear latent evolution may be especially attractive when sparse sensing and filtering are part of the task. Beyond discrete-time Koopman and DMD formulations, continuous-time latent dynamics have been explored through neural ordinary differential equations \cite{chen2018neural} and latent ODE models for irregularly sampled time series \cite{rubanova2019latent}. More directly relevant to the sensing problem, Peyron et al.\ \cite{peyron2021latent} demonstrated that data assimilation can be performed entirely in the latent space of a convolutional autoencoder, achieving competitive reconstruction at reduced cost. However, none of these latent-space DA approaches have systematically compared how the choice between linear and nonlinear latent dynamics affects assimilation quality under sparse, noisy observations---precisely the question this work addresses. At the same time, classical linear stochastic estimation highlights an important baseline: static linear maps can reconstruct flow fields from sparse measurements, but they do not propagate state uncertainty or temporal memory in the way a dynamical filter does \cite{adrian1994stochastic}.

This work addresses a task-driven question motivated by sparse sensing: when the downstream objective is data assimilation, which choice of latent dynamics model yields the most effective state reconstruction? We study this empirically using an encoder-decoder backbone across four benchmarks: 1D viscous Burgers, 2D flow past a cylinder, 2D dambreak, and 3D flow past a sphere, and we compare a parameter-conditioned Koopman latent model against three alternatives: parametric DMD, an MLP latent model, and an xLSTM latent model. Our contributions are threefold. (i) We introduce Deep Koopman Sensing, a unified latent-space formulation that combines nonlinear representation learning, learned latent dynamics, and sparse sensing. (ii) We demonstrate that open-loop rollout accuracy is not a reliable surrogate for sensing performance: models that are competitive in forecasting can exhibit markedly worse behavior once filtering closes the loop. (iii) Across the primary experiments and additional ablation studies, Koopman latent dynamics provide the most robust and consistently low assimilation error under sparse and noisy observations. An overview of the training and sensing workflow is shown in \Cref{fig:main_fig}.

\begin{figure}[!htbp]
    \centering
    \includegraphics[width=0.97\linewidth]{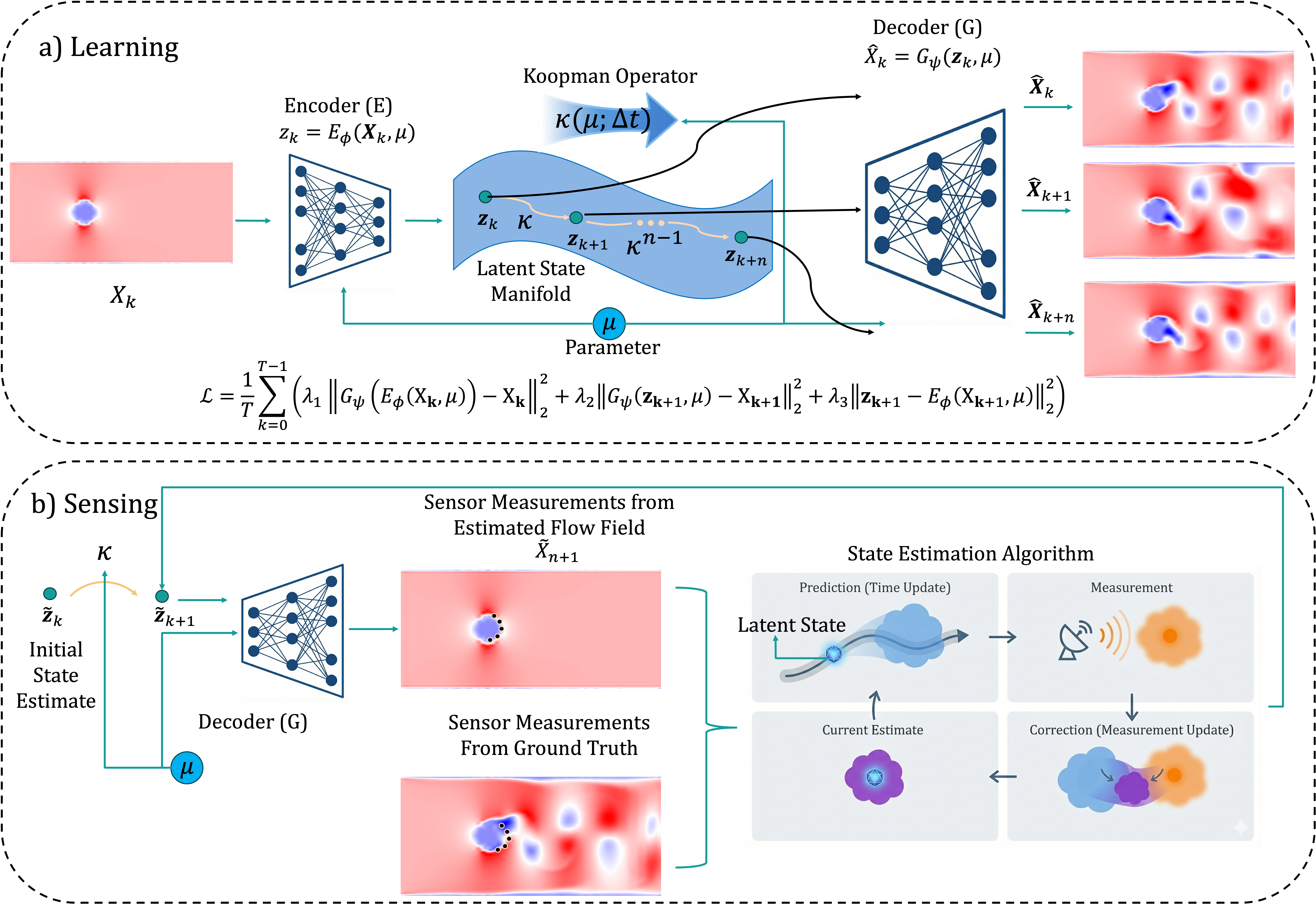}
    \caption{Overview of Deep Koopman Sensing. (a) Learning phase: joint training of the encoder, parameter-conditioned latent dynamics, and decoder. (b) Sensing phase: latent-space sensing and state estimation, where sparse measurements are compared against decoded sensor values and used to update the latent state.}
    \label{fig:main_fig}
\end{figure}
\section{Deep Koopman Sensing}\label{sec:method}

We seek a reduced-order representation that supports both forecasting and online correction from sparse observations. Let $\mathbf{x}_k \in \mathbb{R}^{N}$ denote the full state at time $t_k$, $\mathbf{z}_k \in \mathbb{R}^{r}$ the latent state with $r \ll N$, $\mu\in\mathbb{R}$ the physical parameter, and $\mathbf{y}_k \in \mathbb{R}^{m}$ the vector of sensor measurements. The learned neural baselines (Koopman, MLP, and xLSTM) share the same encoder-decoder backbone, whereas pDMD is included as a separate POD/DMD reduced-order baseline.

\subsection{Problem formulation}\label{subsec:problem_formulation}

The full state is generated by an autonomous, parameterized dynamical system arising from the spatial discretization of a partial differential equation. Let $\mathbf{u}(\cdot,t)$ denote the continuous flow state, governed by
\begin{equation}
    \partial_t \mathbf{u}=\mathcal N(\mathbf{u};\mu)
    \label{eq:governing_pde}
\end{equation}
on a spatial domain with prescribed boundary and initial conditions, where $\mathcal N$ is a nonlinear spatial operator and $\mu$ is a physical parameter such as a Reynolds number or a kinematic viscosity. Spatial discretization on $N$ degrees of freedom yields the autonomous system
\begin{equation}
    \frac{d\mathbf{x}}{dt}=\mathbf{F}(\mathbf{x};\mu),
    \qquad \mathbf{x}(t)\in\mathbb{R}^{N},
    \label{eq:semi_discrete}
\end{equation}
and uniform sampling with time step $\Delta t$ defines the discrete flow map
\begin{equation}
    \mathbf{x}_{k+1}=\Phi_{\Delta t}(\mathbf{x}_k;\mu),
    \qquad \mathbf{x}_k:=\mathbf{x}(k\Delta t),
    \label{eq:flow_map}
\end{equation}
which is deterministic for fixed $\mu$. The four benchmark problems of \cref{sec:data} instantiate \cref{eq:governing_pde} with the 1D viscous Burgers equation, the incompressible Navier--Stokes equations for the flows past a cylinder and past a sphere, and the incompressible two-phase Navier--Stokes equations for the dambreak.

Measurements are sparse, noisy point samples of the state,
\begin{equation}
    \mathbf{y}_k=\mathcal S\,\mathbf{x}_k+\boldsymbol{\eta}_k,
    \label{eq:sensing}
\end{equation}
where $\mathcal S\in\{0,1\}^{m\times N}$ with $m\ll N$ is the sampling matrix that extracts the measured quantity at the sensor locations, and $\boldsymbol{\eta}_k$ is observation noise. The task is sequential state estimation: given trajectories of \cref{eq:flow_map} at training parameter values, and at test time the measurement stream \cref{eq:sensing} together with the value of the physical parameter $\mu$, estimate the full state $\mathbf{x}_k$. All four latent models require $\mu$: the Koopman operator is evaluated as $\mathbf{K}(\mu;\Delta t)$, the MLP and xLSTM take the learned parameter embedding $\mathbf{s}=e_{\omega}(\mu)$ of \cref{eq:re_embed} as an additional input, concatenated with the latent state, and pDMD depends on $\mu$ through its own parameter-dependent construction (\cref{app:pdmd}). The remainder of this section constructs the reduced-order components, and \cref{subsec:sensing} assembles them into a filter.

\subsection{Nonlinear Dimensionality Reduction}\label{subsec:representation}

The encoder maps the high-dimensional field to a low-dimensional latent state,
\begin{equation}
    \mathbf{z}_k = E_{\phi}(\mathbf{x}_k, \mu),
\end{equation}
where $E_{\phi}$ is the nonlinear encoder with learnable parameters $\phi$. The decoder $G_{\psi}$ with learnable parameters $\psi$ reconstructs the physical state $\hat{\mathbf{x}}_k$,
\begin{equation}
    \hat{\mathbf{x}}_k = G_{\psi}(\mathbf{z}_k,\mu).
\end{equation}
The decoder is conditioned on the operating parameter in the same way as the encoder, through
the residual parameter embedding of \cref{eq:re_embed}, as shown in \cref{fig:main_fig}. To keep the
expressions that follow readable we write $G_{\psi}(\mathbf{z}_k)$ wherever the parameter argument is
held fixed. The implemented decoder always receives $\mu$.

The encoder-decoder pair is trained so that $G_{\psi}(E_{\phi}(\mathbf{x}_k,\mu))$ approximates $\mathbf{x}_k$ on the trajectory data. During filtering, the analysis state $\hat{\mathbf{z}}_k^{a}$ is mapped back to the physical domain through the same decoder,
\begin{equation}
    \hat{\mathbf{x}}_k^{a}=G_{\psi}(\hat{\mathbf{z}}_k^{a}).
\end{equation}

For parametric problems, the physical parameter is embedded into the latent space through a small neural network,
\begin{equation}
    \mathbf{s}=e_{\omega}(\mu) \in \mathbb{R}^{r},
    \label{eq:re_embed}
\end{equation}
where $e_{\omega}$ is an MLP. In the Koopman model, $\mathbf{s}$ conditions the latent generator. In the dataset-specific autoencoders, the same embedding is also added as a residual to the latent code so that the representation can vary smoothly with the operating parameter. Each model learns its own embedding network jointly with its encoder and decoder, so the embedding has the same form but not the same weights across models.
\subsection{Latent dynamics models}\label{subsec:latent_models}

For the Koopman and MLP baselines, the latent dynamics are written as one-step Markov transitions of the form
\begin{equation}
    \hat{\mathbf{z}}_{k+1} = \mathcal{F}_{\theta}(\hat{\mathbf{z}}_k,\mu).
    \label{eq:latent_markov}
\end{equation}
The xLSTM baseline shares this one-step form, differing from the MLP only in the gating used internally (\cref{app:xlstm}). The key experimental question is how the choice of $\mathcal{F}_{\theta}$ affects data assimilation.
\subsubsection{Koopman latent model}\label{subsec:koopman_model}

In the Koopman formulation, the latent dynamics are linear,
\begin{equation}
    \hat{\mathbf{z}}_{k+1}=\mathbf{K}(\mu;\Delta t)\,\hat{\mathbf{z}}_k,
    \label{eq:koopman}
\end{equation}
where $\mathbf{K}(\mu;\Delta t)\in\mathbb{R}^{r\times r}$ is a parameter-dependent Koopman operator. Rather than learning $K$ directly, we parameterize a continuous-time generator $\mathbf{A}(\mu)$ and obtain the discrete-time operator through the matrix exponential,
\begin{equation}
    \mathbf{K}(\mu;\Delta t)=e^{\mathbf{A}(\mu)\,\Delta t}.
\end{equation}

The generator is built from conservative and dissipative components. Using the parameter embedding $\mathbf{s}=e_{\omega}(\mu)$, two MLPs produce
\begin{align}
    \mathbf{q}(\mathbf{s}) &= Q_{\theta}(\mathbf{s})\in\mathbb{R}^{r^2},\\
    \mathbf{d}(\mathbf{s}) &= d_{\theta}(\mathbf{s})\in\mathbb{R}^{r},
\end{align}
Reshaping $\mathbf{q}(\mathbf{s})$ into a matrix $\mathbf{Q}(\mathbf{s})\in\mathbb{R}^{r\times r}$, we define
\begin{equation}
    \mathbf{A}(\mu) = \mathbf{Q}(\mathbf{s})-\mathbf{Q}(\mathbf{s})^{\top} - \diag\!\big(\mathbf{d}(\mathbf{s}) \circ \mathbf{d}(\mathbf{s})\big).
    \label{eq:generator}
\end{equation}
The skew-symmetric part captures conservative latent rotations, while the negative diagonal term
imposes dissipation. This produces a parameter-aware linear latent model with a stable structure that can be used directly in the prediction step of the Kalman filter. The skew-symmetric-plus-dissipative form makes $\mathbf{K}(\mu;\Delta t)$ non-expansive by construction; we state this property, and test it against the baselines, in \cref{subsec:idealized_filtering_perspective}. The MLP hyperparameters used in the Koopman generator are summarized in \cref{tab:koopman_mlp_hparams}.

\begin{table}[!htbp]
    \centering
    \small
    \caption{Hyperparameters of the MLP components used in the Koopman latent dynamics.}
    \label{tab:koopman_mlp_hparams}
    \renewcommand{\arraystretch}{1.2}
    \setlength{\tabcolsep}{4pt}
    \begin{tabular}{p{3.2cm} p{3.1cm} p{5.0cm}}
        \toprule
        \textbf{Component} & \textbf{Architecture / output} & \textbf{Hyperparameters} \\
        \midrule
        Conservative map $Q_{\theta}(\mathbf{s})$
        & MLP: $\mathbb{R}^{r}\to\mathbb{R}^{r^2}$
        & Hidden layers: $L_Q=4$; hidden width: $r$; activation: GELU; output normalization: LayerNorm (no affine). \\

        Dissipation map $d_{\theta}(\mathbf{s})$
        & MLP: $\mathbb{R}^{r}\to\mathbb{R}^{r}$
        & Hidden layers: $L_d=4$; hidden width: $r$; activation: GELU; output normalization: LayerNorm (no affine). \\
        \bottomrule
    \end{tabular}
\end{table}
\subsubsection{Baseline latent models}\label{subsubsec:baselines}

We compare the Koopman model against three baselines: a multilayer perceptron (MLP) and an xLSTM~\cite{beck2024xlstm}, which share the same encoder–decoder backbone and 
differ only in their latent dynamics, and a parametric DMD (pDMD) model, which instead uses a global POD basis and therefore differs in the representation
as well. The pDMD model retains linear reduced dynamics; the MLP is a fully connected nonlinear Markov transition; and the xLSTM is an xLSTM-gated nonlinear one-step map, which in this work is evaluated with a context length of one, so its memory, normalization and stabilization variables are re-initialized at every step and it too acts as a Markov transition on the latent state. The set therefore comprises two linear transitions and two nonlinear Markov transitions that differ in their internal parameterization. They enter the filter (\cref{subsec:sensing}) through different latent Jacobians $\mathbf{F}_k$: a fixed global linear map for Koopman, a constant parameter-specific linear map $\tilde{\mathbf{A}}(\mu)$ for pDMD, refitted online to that parameter's interpolated reduced coefficients (\cref{eq:pdmd_online_operator}), and state-dependent local linearizations for the MLP and xLSTM. The full formulations and hyperparameters are deferred to \cref{app:baselines} (pDMD in \cref{app:pdmd}, MLP in \cref{app:mlp}, and xLSTM in \cref{app:xlstm}).
\subsection{Joint training objective}\label{subsec:training_objective}

The Koopman, MLP and xLSTM models are trained jointly with the encoder and decoder using a combination of reconstruction and prediction losses; pDMD is fitted separately and uses a POD basis rather than the learned autoencoder. Let $\mathbf{z}_k=E_{\phi}(\mathbf{x}_k,\mu)$. We use
{\small
\begin{equation}
\mathcal{L}
=
\frac{1}{T}
\sum_{k=0}^{T-1}
\left(
\lambda_1 \| G_{\psi}(E_{\phi}(\mathbf{x}_k,\mu)) - \mathbf{x}_k \|_2^2
+
\lambda_2 \| G_{\psi}(\mathcal{F}_{\theta}(\mathbf{z}_k,\mu)) - \mathbf{x}_{k+1} \|_2^2
+
\lambda_3 \| \mathcal{F}_{\theta}(\mathbf{z}_k,\mu) - \mathbf{z}_{k+1} \|_2^2
\right).
\label{eq:training_loss}
\end{equation}
}
Here $\lambda_1$, $\lambda_2$, and $\lambda_3$ are fixed weights assigned to the reconstruction, decoded prediction, and latent prediction terms, respectively. We use $\lambda_1=\lambda_2=\lambda_3=1$ in all experiments. The principal experiments use one-step training. In the multi-step study the horizon is extended to $T\in\{2,4,8\}$ steps by unrolling the latent map recursively from a single encoding, without reinjecting ground-truth encodings at intermediate steps,
\begin{equation}
    \hat{\mathbf{z}}_{k,0}=E_{\phi}(\mathbf{x}_k,\mu),
    \qquad
    \hat{\mathbf{z}}_{k,j+1}=\mathcal{F}_{\theta}(\hat{\mathbf{z}}_{k,j},\mu),
    \qquad j=0,\dots,T-1,
    \label{eq:multistep_rollout}
\end{equation}
and accumulating the decoded-prediction and latent-prediction terms over all $T$ steps against
$\mathbf{x}_{k+j+1}$ and $E_{\phi}(\mathbf{x}_{k+j+1},\mu)$ respectively, while the reconstruction
term is retained at the initial step only. The accumulated terms are averaged over the horizon, so
the weights $\lambda_1,\lambda_2,\lambda_3$ carry the same values as in the one-step case. The Koopman, MLP, and xLSTM models are trained with this objective, whereas pDMD is fit separately through its POD/DMD pipeline. 
\subsection{Latent-space sensing and state estimation}\label{subsec:sensing}

Sparse measurements are assimilated in the reduced space using the trained forecast model together with the corresponding reduced-to-full reconstruction map. The physical measurements are the point samples of \cref{eq:sensing}; since the full state is not available online, the filter models the map from the latent state to the measurement by decoding and sampling. The observation model is written as
\begin{equation}
    \mathbf{y}_k = h(\mathbf{z}_k)+\tilde{\boldsymbol{\eta}}_k,
\end{equation}
where the observation disturbance $\tilde{\boldsymbol{\eta}}_k$ is modeled by the filter as zero-mean Gaussian with covariance $\boldsymbol{\Sigma}_{\eta}$; its exact composition, measurement noise plus sensor-sampled reconstruction error, is given in \cref{rem:dks_instantiation}. For the neural baselines,
\begin{equation}
    h(\mathbf{z}_k)=\mathcal{S}\big(G_{\psi}(\mathbf{z}_k)\big),
\end{equation}
whereas for pDMD the reconstruction is given by the first $r$ POD basis $\mathbf{U}_r$,
\begin{equation}
    h(\mathbf{z}_k)=\mathcal{S}(\mathbf{U}_r \mathbf{z}_k).
\end{equation}

We use an Extended Kalman Filter (EKF) in the main experiments. The latent prediction step is
\begin{align}
    \hat{\mathbf{z}}_{k+1}^{f} &= \mathcal{F}_{\theta}(\hat{\mathbf{z}}_k^{a},\mu),\\
    \mathbf{P}_{k+1}^{f} &= \mathbf{F}_k \mathbf{P}_k^{a} \mathbf{F}_k^{\top}+\boldsymbol{\Sigma}_{\xi},
\end{align}
where $\mathbf{F}_k=\partial \mathcal{F}_{\theta}/\partial \mathbf{z}\vert_{\mathbf{z}=\hat{\mathbf{z}}_k^{a}}$ is the Jacobian of the latent transition and $\boldsymbol{\Sigma}_{\xi}$ is the process-noise covariance. After receiving $\mathbf{y}_{k+1}$ from measurement, the update step reads
\begin{align}
    \mathcal{K}_{k+1}
    &= \mathbf{P}_{k+1}^{f}\mathbf{J}_{k+1}^{\top}\left(\mathbf{J}_{k+1}\mathbf{P}_{k+1}^{f}\mathbf{J}_{k+1}^{\top}+\boldsymbol{\Sigma}_{\eta}\right)^{-1},\\
    \hat{\mathbf{z}}_{k+1}^{a}
    &= \hat{\mathbf{z}}_{k+1}^{f} + \mathcal{K}_{k+1}\big(\mathbf{y}_{k+1}-h(\hat{\mathbf{z}}_{k+1}^{f})\big),\\
    \mathbf{P}_{k+1}^{a}
    &= \left(\mathbf{I}-\mathcal{K}_{k+1}\mathbf{J}_{k+1}\right)\mathbf{P}_{k+1}^{f},
\end{align}
where $\mathbf{J}_{k+1}=\partial h/\partial \mathbf{z}\vert_{\mathbf{z}=\hat{\mathbf{z}}_{k+1}^{f}}$ is the Jacobian of the sensing map. For Koopman, MLP, and pDMD, the filter state has dimension $r$, so $\mathbf{F}_k\in\mathbb{R}^{r\times r}$, $\mathbf{P}_k^{a},\mathbf{P}_k^{f}\in\mathbb{R}^{r\times r}$, $\mathbf{J}_{k+1}\in\mathbb{R}^{m\times r}$, and $\mathcal{K}_{k+1}\in\mathbb{R}^{r\times m}$. The xLSTM baseline is filtered on the same $r$-dimensional latent state: its cell is evaluated over a single timestep with the recurrent variables initialized afresh at each call, so the realized transition is a map on $\mathbb{R}^{r}$ and the shapes above apply unchanged to all four models (\cref{app:xlstm}).

The distinction between latent models enters directly through the choice of $\mathbf{F}_k$ and observation $h$. For the Koopman model, $\mathbf{F}_k=\mathbf{K}(\mu;\Delta t)$ is a global linear map, while $h$ can be generally nonlinear. For the pDMD baseline, the  underlying DMD forecast is linear in the stacked POD coefficient space, but evaluation at a new parameter additionally involves interpolation across parameter space. For the MLP and xLSTM baselines the forecast map is nonlinear and its Jacobian $\mathbf{F}_k$ depends on the current latent state, so the EKF relies on a local linearization of a nonlinear latent transition.

\section{Datasets and Experimental Setup}\label{sec:data}

We evaluate the framework on four benchmark problems spanning one, two, and three spatial dimensions: 1D Viscous Burgers, 2D flow past a circular cylinder, 2D dambreak, 3D flow past a sphere. The 2D and 3D sensor layouts are shown in \cref{fig:sensor_layouts}. For the 2D and 3D benchmarks, the data is generated using CFD, and the fields are subsequently interpolated onto the uniform Cartesian grid on which the models are trained and evaluated. The resolutions quoted below ($64\times64$ in 2D and $64^{3}$ in 3D) are those of this resampling grid rather than of the simulation mesh, and all errors reported in \cref{sec:results} are computed on the resampled fields. Unless noted otherwise, the main comparison uses latent dimension $r=32$ and measurement noise equal to $20\%$ of the maximum measured quantity in the corresponding test trajectory.

\begin{figure}[!htbp]
    \centering
    \begin{subfigure}[!htbp]{0.32\linewidth}
        \centering
        \includegraphics[width=\linewidth]{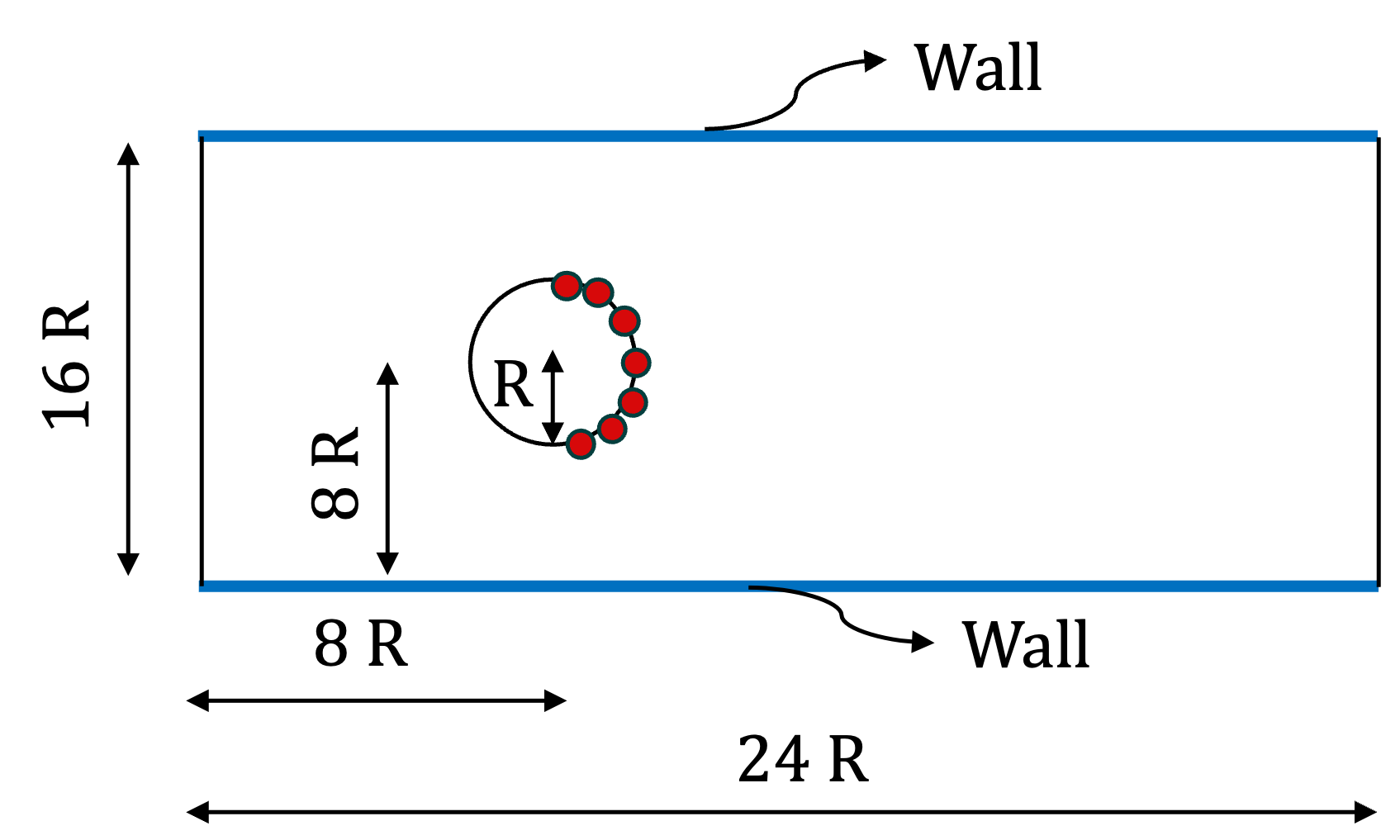}
        \caption{2D cylinder.}
        \label{fig:sensor_layout_cyl}
    \end{subfigure}\hfill
    \begin{subfigure}[!htbp]{0.32\linewidth}
        \centering
        \includegraphics[width=\linewidth]{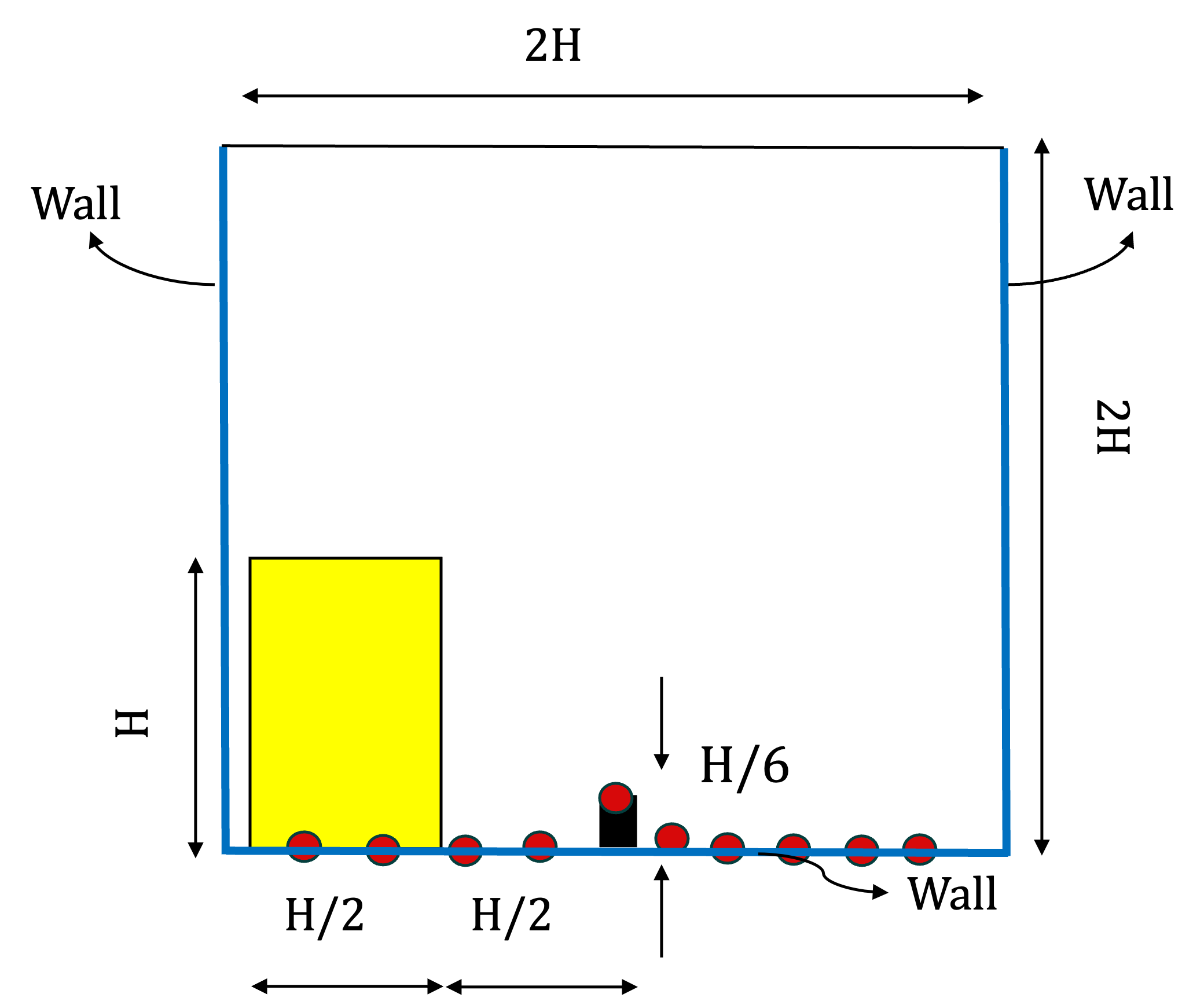}
        \caption{2D dambreak.}
        \label{fig:sensor_layout_dambreak}
    \end{subfigure}\hfill
    \begin{subfigure}[!htbp]{0.32\linewidth}
        \centering
        \includegraphics[width=\linewidth]{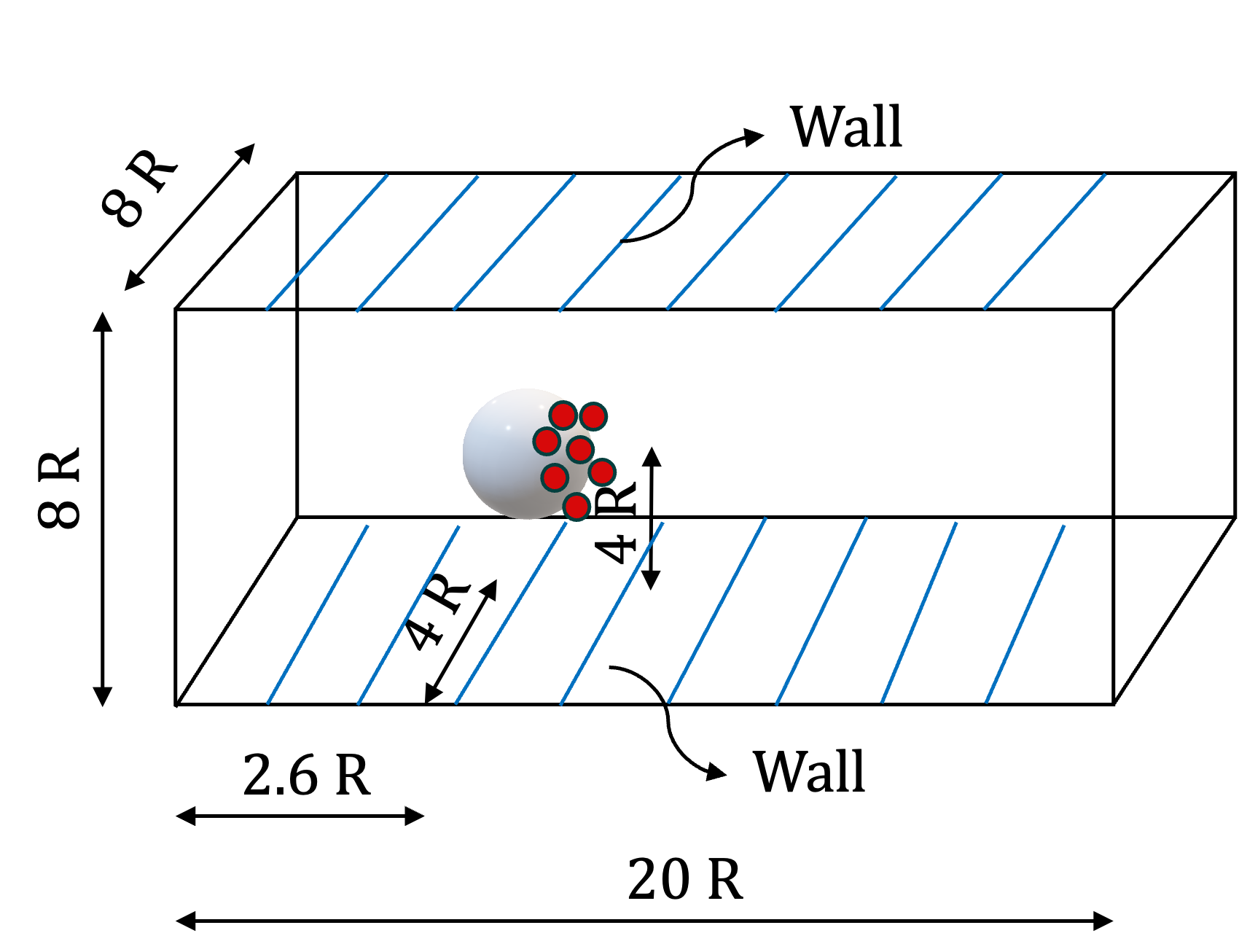}
        \caption{3D sphere.}
        \label{fig:sensor_layout_sphere}
    \end{subfigure}
    \caption{The computational domains used to generate the cylinder, dambreak, and sphere datasets. For the cylinder and sphere cases, the top and bottom surfaces are walls. For the dambreak case, the sides and bottom surfaces are walls. The DA sensor locations are also shown. In the cylinder and sphere cases the sensors lie in the wake close to the body surface, and in the dambreak case along the bottom surface and around the obstacle. The sensors are symmetrically and equidistantly arranged.}
    \label{fig:sensor_layouts}
\end{figure}
\subsection{1D viscous Burgers equation}

The 1D viscous Burgers system is defined by
\begin{equation}
\frac{\partial u}{\partial t} + u \frac{\partial u}{\partial x} = \nu \frac{\partial^2 u}{\partial x^2},
\end{equation}
on the periodic domain $x\in[0,2\pi)$, discretized with $N_x=1024$ uniformly spaced points. Every trajectory starts from the same two-mode profile $u_0(x)=\sin x+\tfrac{1}{2}\cos 2x$, so the trajectories differ only through the viscosity. The reference solutions are computed with an explicit finite-difference scheme, first-order upwind for advection and second-order central differences for diffusion, with the step limited by an advective Courant number of $0.1$ and a diffusive stability factor of $0.45$. Each run is integrated to $t=0.5$ and stored as 100 frames spaced about $5\times10^{-3}$ apart. The time step $\Delta t=0.005$ in \cref{tab:vb_koopman_architecture} is the increment at which the learned latent operator $\mathbf{K}(\mu;\Delta t)$ is evaluated, same as the sampling $\Delta t$. We generate 50 trajectories of length 100 time steps for training, with viscosity sampled uniformly from $\nu\in[0.01,1.0]$, and reserve 10 unseen trajectories for testing. Sensors are placed equidistantly over the spatial domain. The Koopman encoder-decoder architecture and training hyperparameters used for this dataset are summarized in \Cref{tab:vb_koopman_architecture}.

\begin{table}[!htbp]
    \centering
    \scriptsize
    \caption{Neural-network architecture and training hyperparameters for the 1D viscous Burgers Koopman model. Here $r$ denotes the latent dimension.}
    \label{tab:vb_koopman_architecture}
    \renewcommand{\arraystretch}{1.2}
    \begin{adjustbox}{max width=\textwidth}
    \begin{tabular}{p{3.2cm} p{3.2cm} p{3.2cm} p{3.2cm}}
        \toprule
        \multicolumn{4}{c}{\textbf{Encoder / Decoder Architecture}} \\
        \midrule
        \multicolumn{2}{c}{\textbf{Encoder}} &
        \multicolumn{2}{c}{\textbf{Decoder}} \\
        \midrule
        \textbf{Component} & \textbf{Details} &
        \textbf{Component} & \textbf{Details} \\
        \midrule
        Input &
        1D field, size $1024$ (single channel) &
        Input &
        Latent state ($r$-dimensional) \\

        Conv layers &
        3$\times$Conv1D: $1\!\rightarrow\!8\!\rightarrow\!16\!\rightarrow\!32$, kernel $=3$, stride $=2$, padding $=1$; GELU after first two convs &
        FC layers &
        $r \rightarrow 256 \rightarrow 512 \rightarrow (32\times128)$ with GELU between FC layers \\

        Flatten \& projection &
        Flatten $\rightarrow$ FC: $(32\times128)\rightarrow512\rightarrow256\rightarrow r$ (GELU between FC layers) &
        Reshape &
        $(32\times128)\rightarrow (32,128)$ \\

        Residual term &
        Adds parameter embedding: $\mathbf{z}=E_\phi(\mathbf{x})+\mathbf{s}$ &
        Upsampling &
        3$\times$ linear interpolation, scale factor $2$ (mode = linear) \\

        Output &
        Latent code $\mathbf{z}\in\mathbb{R}^{r}$ &
        Deconv layers &
        ConvTranspose1D (stride $=1$, padding $=1$, kernel $=3$): $32\!\rightarrow\!16\!\rightarrow\!8\!\rightarrow\!1$; GELU after first two deconvs \\

        \midrule
        \multicolumn{4}{c}{\textbf{Training Hyperparameters}} \\
        \midrule
        \multicolumn{2}{c}{\textbf{Parameter}} &
        \multicolumn{2}{c}{\textbf{Value}} \\
        \midrule
        \multicolumn{2}{c}{Training steps} &
        \multicolumn{2}{c}{1000} \\

        \multicolumn{2}{c}{Batch size} &
        \multicolumn{2}{c}{256} \\

        \multicolumn{2}{c}{Optimizer} &
        \multicolumn{2}{c}{Adam} \\

        \multicolumn{2}{c}{Learning rate} &
        \multicolumn{2}{c}{$1\times 10^{-4}$} \\

        \multicolumn{2}{c}{Latent time step} &
        \multicolumn{2}{c}{$\Delta t = 0.005$} \\
        \bottomrule
    \end{tabular}
    \end{adjustbox}
\end{table}
\subsection{2D flow past a circular cylinder}

The 2D cylinder dataset consists of incompressible Navier-Stokes solutions resampled onto a uniform $64\times64$ grid. The domain spans $x\in[-8R,16R]$ and $y\in[-4R,4R]$, with a cylinder of radius $R=1$ centred at the origin. Each trajectory contains 100 time steps. The training set contains 90 trajectories, the test set 11 trajectories, and the Reynolds number is sampled over $\mathrm{Re}\in[59.5,990.5]$. Test trajectories are drawn from the same parameter range but are not seen during training. The data are generated using OpenFOAM. Sensors are placed on the cylinder surface and in its wake region, as shown in \cref{fig:sensor_layout_cyl}. The Koopman encoder-decoder architecture and training hyperparameters used for this dataset are summarized in \Cref{tab:cyl_koopman_architecture}.

\begin{table}[!htbp]
    \centering
    \scriptsize
    \caption{Neural-network architecture and training hyperparameters for the 2D cylinder Koopman model. Here $r$ denotes the latent dimension. The four channels of inputs refers to velocities ($u,v$), pressure ($p$) and the cell type (boundary, cylinder, inlet, outlet, interior).}
    \label{tab:cyl_koopman_architecture}
    \renewcommand{\arraystretch}{1.2}
    \begin{adjustbox}{max width=\textwidth}
    \begin{tabular}{p{3.2cm} p{3.2cm} p{3.2cm} p{3.2cm}}
        \toprule
        \multicolumn{4}{c}{\textbf{Encoder / Decoder Architecture}} \\
        \midrule
        \multicolumn{2}{c}{\textbf{Encoder}} &
        \multicolumn{2}{c}{\textbf{Decoder}} \\
        \midrule
        \textbf{Component} & \textbf{Details} &
        \textbf{Component} & \textbf{Details} \\
        \midrule
        Input & 2D field, $64\times64$, 4 channels   &
        Input & Koopman latent state ($r$-dimensional) \\

        Conv layers &
        3$\times$Conv2D: $4\!\rightarrow\!128\!\rightarrow\!256\!\rightarrow\!512$; kernel $=5$, stride $=2$, padding $=2$; GELU after first two convs &
        FC layers &
        $r \rightarrow 256 \rightarrow 512 \rightarrow (512\times8\times8)$ with GELU between FC layers \\

        Flatten / bottleneck &
        Flatten $\rightarrow$ FC: $(512\times8\times8)\rightarrow512\rightarrow256\rightarrow r$ (GELU between FC layers) &
        Reshape &
        $(512\times8\times8) \rightarrow (512,8,8)$ \\

        Residual term &
        Adds parameter embedding: $\mathbf{z}=E_\phi(\mathbf{x})+\mathbf{s}$ &
        Upsampling &
        3$\times$ bilinear interpolation, scale factor $2$ \\

        Output &
        Latent code $\mathbf{z}\in\mathbb{R}^{r}$ &
        Deconv layers &
        ConvTranspose2D (stride $=1$, padding $=2$, kernel $=5$): $512\!\rightarrow\!256\!\rightarrow\!128\!\rightarrow\!\blue{3}$; GELU after first two deconvs \\

        \midrule
        \multicolumn{4}{c}{\textbf{Training Hyperparameters}} \\
        \midrule
        \multicolumn{2}{c}{\textbf{Parameter}} &
        \multicolumn{2}{c}{\textbf{Value}} \\
        \midrule
        \multicolumn{2}{c}{Training steps} &
        \multicolumn{2}{c}{1000} \\

        \multicolumn{2}{c}{Batch size} &
        \multicolumn{2}{c}{256} \\

        \multicolumn{2}{c}{Optimizer} &
        \multicolumn{2}{c}{Adam} \\

        \multicolumn{2}{c}{Learning rate} &
        \multicolumn{2}{c}{$1\times 10^{-4}$} \\

        \multicolumn{2}{c}{Time step} &
        \multicolumn{2}{c}{$\Delta t = 0.5$} \\
        \bottomrule
    \end{tabular}
    \end{adjustbox}
\end{table}
\subsection{2D dambreak}

The 2D dambreak dataset represents a two-phase flow resampled onto a uniform $64\times64$ grid. As in the cylinder case, each trajectory contains 100 time steps, with 90 trajectories used for training and 11 reserved for testing. The simulations are parameterized by the kinematic viscosity of the liquid phase,
$\nu_{w}\in[1.0\times10^{-3},\,9.9\times10^{-2}]\,\mathrm{m}^{2}\mathrm{s}^{-1}$ over the 90 training
cases. Test cases are sampled from the same uniform parameter sweep but are not seen during training. The data are generated using OpenFOAM and follow the canonical dambreak geometry shown in \cref{fig:sensor_layout_dambreak}. Sensors are placed along the bottom boundary and around the obstacle. The encoder-decoder backbone follows the same 2D convolutional structure as the cylinder case
(\cref{tab:cyl_koopman_architecture}), with six input and five output channels. 
\subsection{3D flow past a sphere}

The 3D sphere dataset consists of incompressible Navier-Stokes simulations of unsteady flow past a fixed sphere, resampled onto a uniform $64\times64\times64$ grid. Each trajectory contains 100 time steps. The training set comprises 90 trajectories and the test set 11 trajectories. The Reynolds number is sampled over $\mathrm{Re}\in[59.5,990.5]$, and test trajectories are again held out from training. Sensors are distributed on the sphere surface and in the wake region, as shown in \cref{fig:sensor_layout_sphere}. The Koopman encoder-decoder architecture and training hyperparameters used for this dataset are summarized in \Cref{tab:sphere_koopman_architecture}.

\begin{table}[!htbp]
    \centering
    \scriptsize
    \caption{Neural-network architecture and training hyperparameters for the 3D sphere Koopman model. Here $r$ denotes the latent dimension. The five channels of inputs refers to velocities ($u, v, w$), pressure ($p$) and the cell type (boundary, sphere, inlet, outlet, interior).}
    \label{tab:sphere_koopman_architecture}
    \renewcommand{\arraystretch}{1.2}
    \begin{adjustbox}{max width=\textwidth}
    \begin{tabular}{p{3.2cm} p{3.2cm} p{3.2cm} p{3.2cm}}
        \toprule
        \multicolumn{4}{c}{\textbf{Encoder / Decoder Architecture}} \\
        \midrule
        \multicolumn{2}{c}{\textbf{Encoder}} &
        \multicolumn{2}{c}{\textbf{Decoder}} \\
        \midrule
        \textbf{Component} & \textbf{Details} &
        \textbf{Component} & \textbf{Details} \\
        \midrule
        Input &
        3D field, $64\times64\times64$, 5 channels &
        Input &
        Koopman latent state ($r$-dimensional) \\

        Conv layers &
        4$\times$Conv3D: $5\!\to\!64\!\to\!128\!\to\!256\!\to\!512$; kernel $3$, stride $2$, padding $1$, GELU &
        FC layers &
        $r \to 1024 \to (512\!\times\!4\!\times\!4\!\times\!4)$ \\

        Flatten / bottleneck &
        Flatten $\to$ FC: $(512\!\times\!4\!\times\!4\!\times\!4)\to1024\to r$ &
        Reshape &
        $(512\!\times\!4\!\times\!4\!\times\!4)\to [512,4,4,4]$ \\

        Activation &
        GELU (all encoder Convolutional/Fully Connected layers) &
        Deconvolution &
        4$\times$ConvTranspose3D: $512\!\to\!256\!\to\!128\!\to\!64\!\to\!5$; kernel $3$, stride $2$, padding $1$, output\_padding $1$; GELU except last \\

        Residual term &
        Adds parameter embedding to latent code &
        Output &
        3D field, $64\times64\times64$, 5 channels \\
        \midrule
        \multicolumn{4}{c}{\textbf{Training Hyperparameters}} \\
        \midrule
        \multicolumn{2}{c}{\textbf{Parameter}} &
        \multicolumn{2}{c}{\textbf{Value}} \\
        \midrule
        \multicolumn{2}{c}{Training steps} &
        \multicolumn{2}{c}{40000} \\

        \multicolumn{2}{c}{Batch size} &
        \multicolumn{2}{c}{32} \\

        \multicolumn{2}{c}{Optimizer} &
        \multicolumn{2}{c}{Adam} \\

        \multicolumn{2}{c}{Learning rate} &
        \multicolumn{2}{c}{$1\times 10^{-4}$} \\

        \multicolumn{2}{c}{Time step} &
        \multicolumn{2}{c}{$\Delta t = 0.5$} \\
        \bottomrule
    \end{tabular}
    \end{adjustbox}
\end{table}

\subsection{Evaluation protocol}\label{subsec:eval_protocol}

The principal comparison is carried out with latent dimension $r=32$. The models share the encoder--decoder backbone and have approximately matched parameter counts (\cref{tab:ld32_param_counts}). We do not claim that they are matched in forward-prediction accuracy and the argument of this paper does not require them to be equal. What the comparison holds fixed is the representation family, the parameter budget, the sensors, the noise and the filter settings, leaving the latent dynamics model as the variable of interest. The resulting trained models are then used in the data-assimilation setting. We evaluate the models in two tasks. In the \emph{forward run}, the latent state is propagated using only the learned dynamics. In \emph{data assimilation} (DA), the same learned dynamics are embedded inside the EKF described in \cref{subsec:sensing}. In the main setting, both the forward run and DA start from an approximate latent initial condition built from training data. Measurement noise is Gaussian, with standard deviation equal to $20\%$ of the maximum measured variable in the test trajectory unless stated otherwise. We use 16 sensors for all experiments across all datasets. The observed variable is chosen to be the velocity $u$ for 1D Burgers, the pressure $p$ for the 2D cylinder, 2D dambreak and the 3D sphere. The remaining channels are reconstructed but not measured. In the EKF, the process covariance is a small multiple of the identity, $\boldsymbol{\Sigma}_{\xi}=q\mathbf{I}$ with $q=10^{-4}$; the initial covariance is the identity. The measurement covariance is set as the square of the noise level. These choices are kept fixed across all models and experiments. Each latent model is trained once for each benchmark and latent dimension. Within each filter, every model receives the same measurement-noise realization for a given test trajectory. Performance is reported as the mean squared error (MSE) at the final time, averaged over the test trajectories, taken over all reconstructed field channels,

\begin{equation}
    \mathrm{MSE}=\frac{1}{N_{c}C}\sum_{j=1}^{C}\sum_{i=1}^{N_{c}}
    \big(\hat{x}^{(j)}_{i}-x^{(j)}_{i}\big)^{2},
    \label{eq:mse_metric}
\end{equation}
where $N_{c}$ is the number of grid cells, $C$ the number of scored channels, and the sum runs over
the final-time field. The channels entering this average are those scored against the ground truth: one for 1D
Burgers, three for the 2D cylinder (two velocity components and pressure), five for the 2D dambreak
(two velocity components, the two pressure fields $p_{\mathrm{rgh}}$ and $p$, and the water
volume fraction $\alpha$),
and four for the 3D sphere (three velocity components and pressure). Error bars in all bar charts are the sample standard deviation of that quantity across the test trajectories, normalized by the same denominator as the bar itself. In the bar plots, each dataset is normalized by the corresponding forward-run MLP error for the same scenario, so the MLP rollout baseline is 1.0. This normalization makes the gap between open-loop and assimilation performance directly comparable across datasets. We also report the average wall-clock time per DA step.

\begin{table}[t]
\centering
\scriptsize
\setlength{\tabcolsep}{4pt}
\renewcommand{\arraystretch}{1.1}
\begin{adjustbox}{max width=\linewidth}
\begin{tabular}{llccc}
\toprule
Dataset & Model & \# Encoder params & \# Decoder params & \# Latent/Forward params \\
\midrule
\multirow{3}{*}{1D viscous Burgers}
& MLP     & 2{,}243{,}504  & 2{,}243{,}249  & 43{,}607 \\
& xLSTM   & 2{,}243{,}504  & 2{,}243{,}249  & 43{,}104 \\
& Koopman & 2{,}243{,}504  & 2{,}243{,}249  & 43{,}296 \\
\midrule
\multirow{3}{*}{2D cylinder}
& MLP     & 21{,}031{,}264 & 21{,}056{,}003 & 43{,}607 \\
& xLSTM   & 21{,}031{,}264 & 21{,}056{,}003 & 43{,}104 \\
& Koopman & 21{,}031{,}264 & 21{,}056{,}003 & 43{,}296 \\
\midrule
\multirow{3}{*}{2D dambreak}
& MLP     & 21{,}037{,}664 & 21{,}062{,}405 & 43{,}607 \\
& xLSTM   & 21{,}037{,}664 & 21{,}062{,}405 & 43{,}104 \\
& Koopman & 21{,}037{,}664 & 21{,}062{,}405 & 43{,}296 \\
\midrule
\multirow{3}{*}{3D sphere}
& MLP     & 38{,}247{,}008 & 38{,}274{,}949 & 43{,}607 \\
& xLSTM   & 38{,}247{,}008 & 38{,}274{,}949 & 45{,}696 \\
& Koopman & 38{,}247{,}008 & 38{,}274{,}949 & 43{,}296 \\
\bottomrule
\end{tabular}
\end{adjustbox}
\caption{Parameter counts for each learned model and dataset. The Koopman, MLP, and xLSTM models share the same encoder--decoder backbone and use closely matched latent/forward parameter counts.}
\label{tab:ld32_param_counts}
\end{table}
\section{Results and discussion}\label{sec:results}
\subsection{Rollout accuracy does not predict DA quality}

\Cref{fig:main_inexact} summarizes the principal experiments, in which both the forward run and DA start from an approximate initial condition. The two panels rank the models differently. In the forward run the MLP is the better open-loop predictor on three of the four benchmarks, so nonlinear latent dynamics can be effective when no measurements are assimilated, and a model chosen on rollout alone would often be the MLP, a conclusion that would be wrong for sparse sensing. Under DA the ranking changes sharply: the Koopman model yields the lowest assimilation error on all four datasets, often by a wide margin. Assimilation lowers the error of both linear latent models on the benchmarks used here, whereas it degrades the nonlinear models in most cases; the xLSTM on the dambreak is the only case in which a nonlinear model clearly benefits from assimilation relative to its own forecast, and even there it does not outperform the Koopman model. pDMD is a substantially worse open-loop predictor, yet it consistently extracts useful information from the measurements, while the MLP generally does not. This distinction is the most important message of the paper: a latent model should be evaluated in the closed-loop setting for which it is intended, not only in open-loop forecasting.

The same qualitative trend persists when exact initial conditions are used to initialize the DA loop, as shown in \cref{fig:main_exact}. The average wall-clock time per DA step is shown in \cref{fig:main_time} for approximate initialization and in \cref{fig:main_exact_time} for exact initialization. For the Koopman model the forecast Jacobian is free, since $\mathbf{F}_k\equiv\mathbf{K}(\mu;\Delta t)$ is constant once the parameter is known and need never be recomputed. The pDMD cost on the 3D sphere is roughly an order of magnitude above every other model and above its own cost on the other three benchmarks; this is not the stacked operator, which \cref{app:pdmd} confines to the offline fit, but the reconstruction $\mathbf{U}_r\hat{\mathbf{z}}$, whose cost scales with the ambient dimension $N$ and is therefore far larger in 3D than in 2D.


\begin{figure}[!htbp]
    \centering

    \begin{subfigure}[!htbp]{\linewidth}
        \centering
        \begin{minipage}[!htbp]{0.025\linewidth}
            \vspace{0pt}
            \raggedright \hspace*{-0.4em}\textbf{(a)}
        \end{minipage}%
        \begin{minipage}[!htbp]{0.895\linewidth}
            \vspace{0pt}
            \centering
            \includegraphics[width=\linewidth]{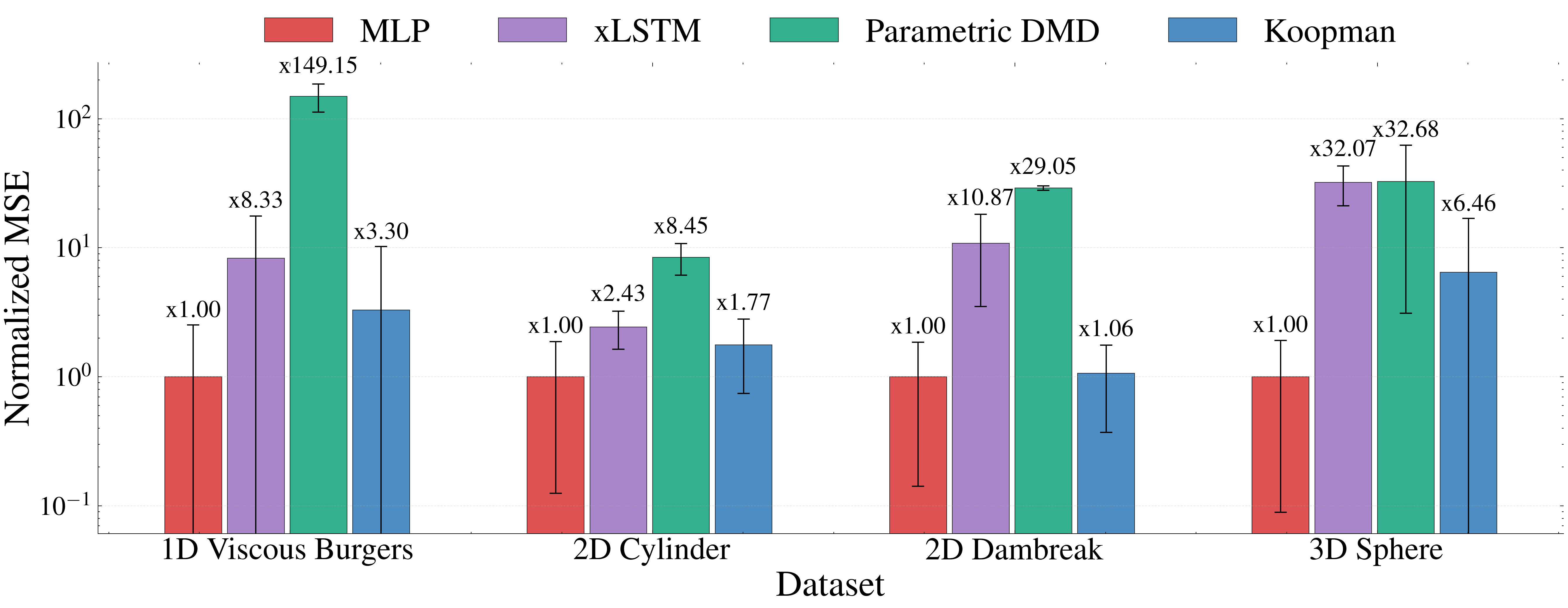}
        \end{minipage}
    \end{subfigure}

    \vspace{0.5em}

    \begin{subfigure}[!htbp]{\linewidth}
        \centering
        \begin{minipage}[!htbp]{0.025\linewidth}
            \vspace{0pt}
            \raggedright \hspace*{-0.4em}\textbf{(b)}
        \end{minipage}%
        \begin{minipage}[!htbp]{0.895\linewidth}
            \vspace{0pt}
            \centering
            \includegraphics[width=\linewidth]{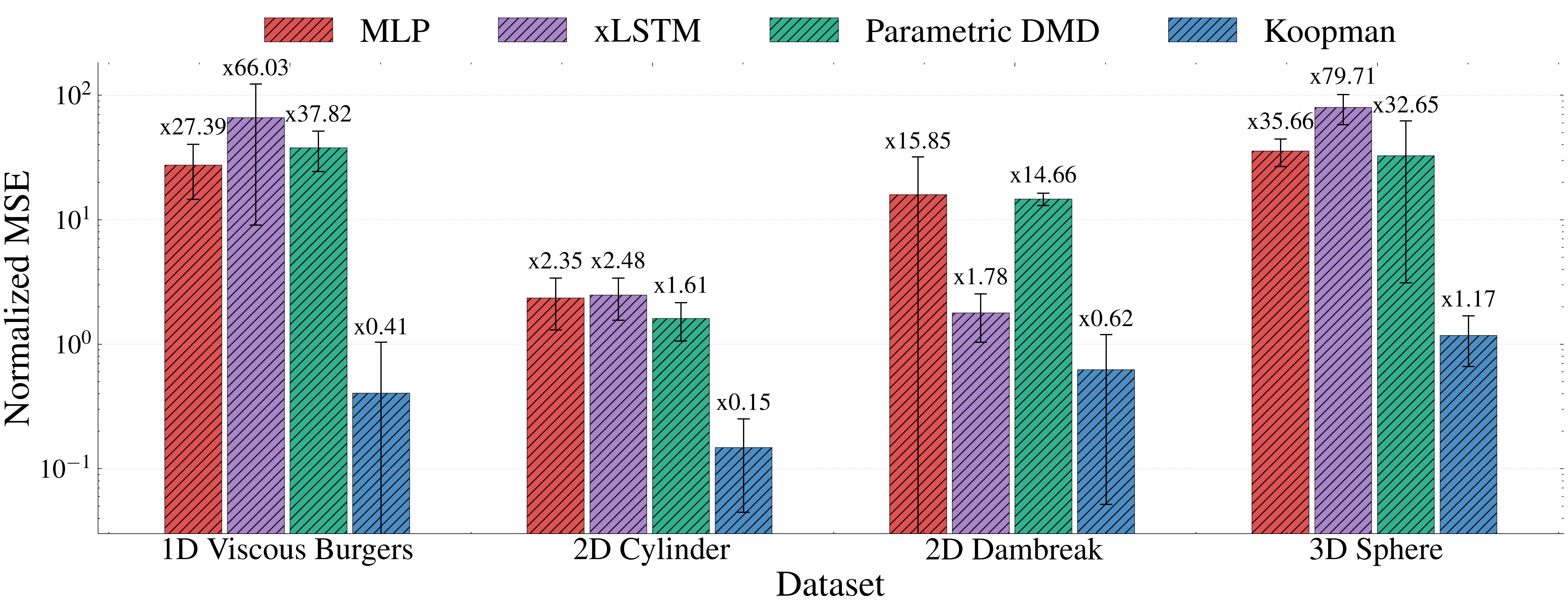}
        \end{minipage}
    \end{subfigure}

    \caption{Comparison of the different latent models in forward-run and DA settings using an approximate initial condition, latent dimension $r=32$, and measurement noise level $20\%$. (a) Forward run. (b) Data assimilation. Errors are normalized by the corresponding forward-run MLP error for each dataset, so the MLP rollout baseline is 1.0. The reported errors are the test-set averages of the final-time MSE. Koopman is not uniformly best in open-loop rollout, but it gives the lowest DA error on all four datasets.}
    \label{fig:main_inexact}
\end{figure}

\begin{figure}[!htbp]
    \centering
    \includegraphics[width=0.92\linewidth]{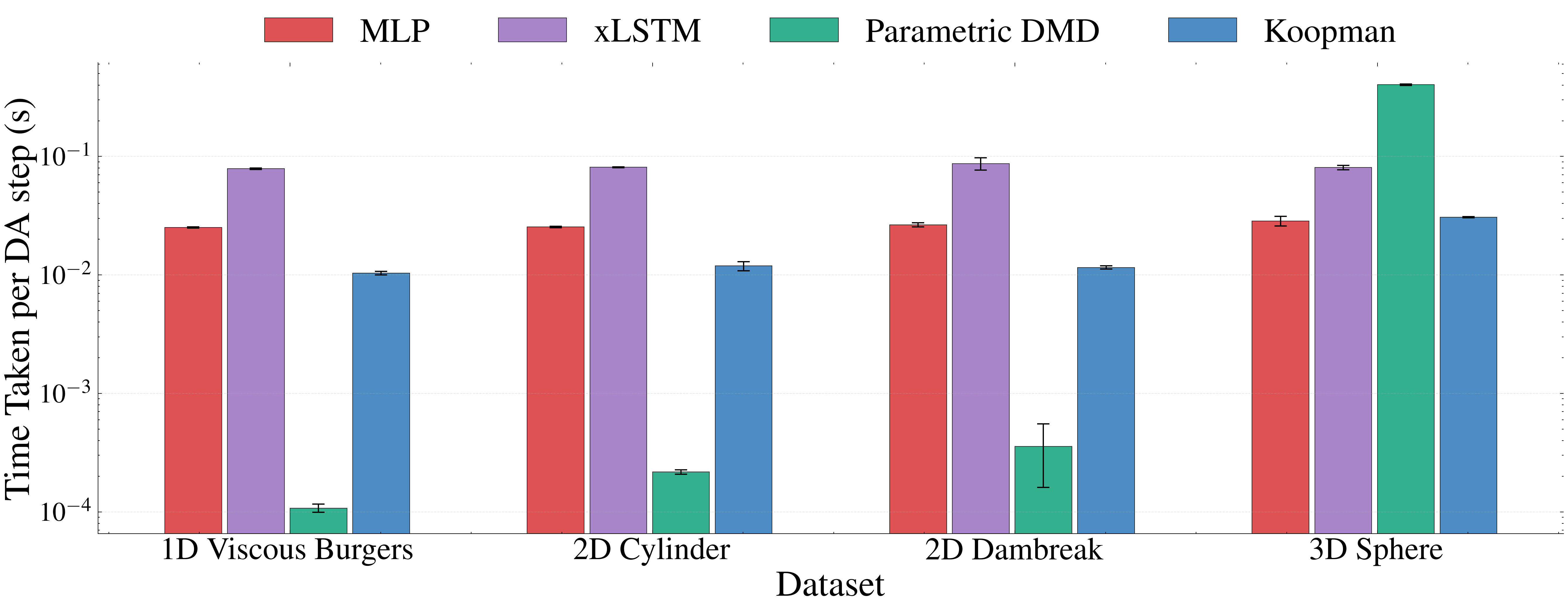}
    \caption{Comparison of the average wall-clock time per DA step for latent dimension $r=32$ and measurement noise level $20\%$. }
    \label{fig:main_time}
\end{figure}

\begin{figure}[!htbp]
    \centering

    \begin{subfigure}[!htbp]{\linewidth}
        \centering
        \begin{minipage}[!htbp]{0.025\linewidth}
            \vspace{0pt}
            \raggedright \hspace*{-0.4em}\textbf{(a)}
        \end{minipage}%
        \begin{minipage}[!htbp]{0.895\linewidth}
            \vspace{0pt}
            \centering
            \includegraphics[width=\linewidth]{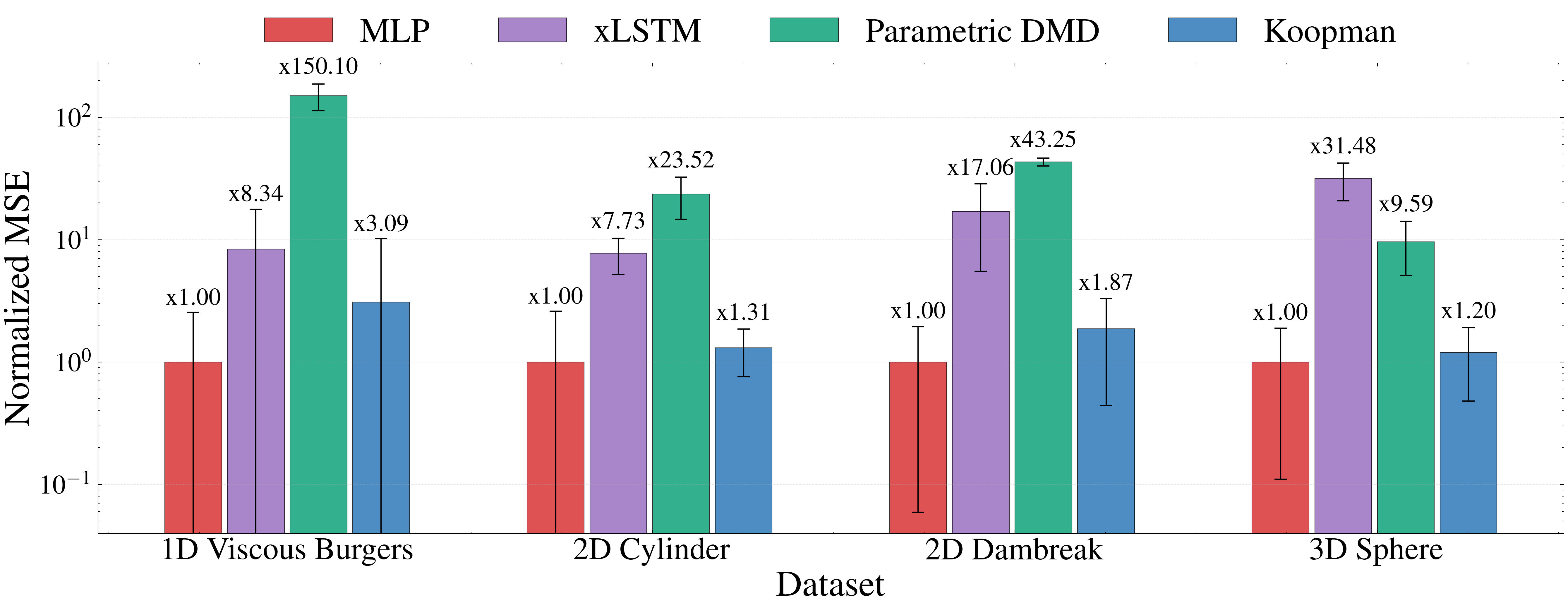}
        \end{minipage}
    \end{subfigure}

    \vspace{0.5em}

    \begin{subfigure}[!htbp]{\linewidth}
        \centering
        \begin{minipage}[!htbp]{0.025\linewidth}
            \vspace{0pt}
            \raggedright \hspace*{-0.4em}\textbf{(b)}
        \end{minipage}%
        \begin{minipage}[!htbp]{0.895\linewidth}
            \vspace{0pt}
            \centering
            \includegraphics[width=\linewidth]{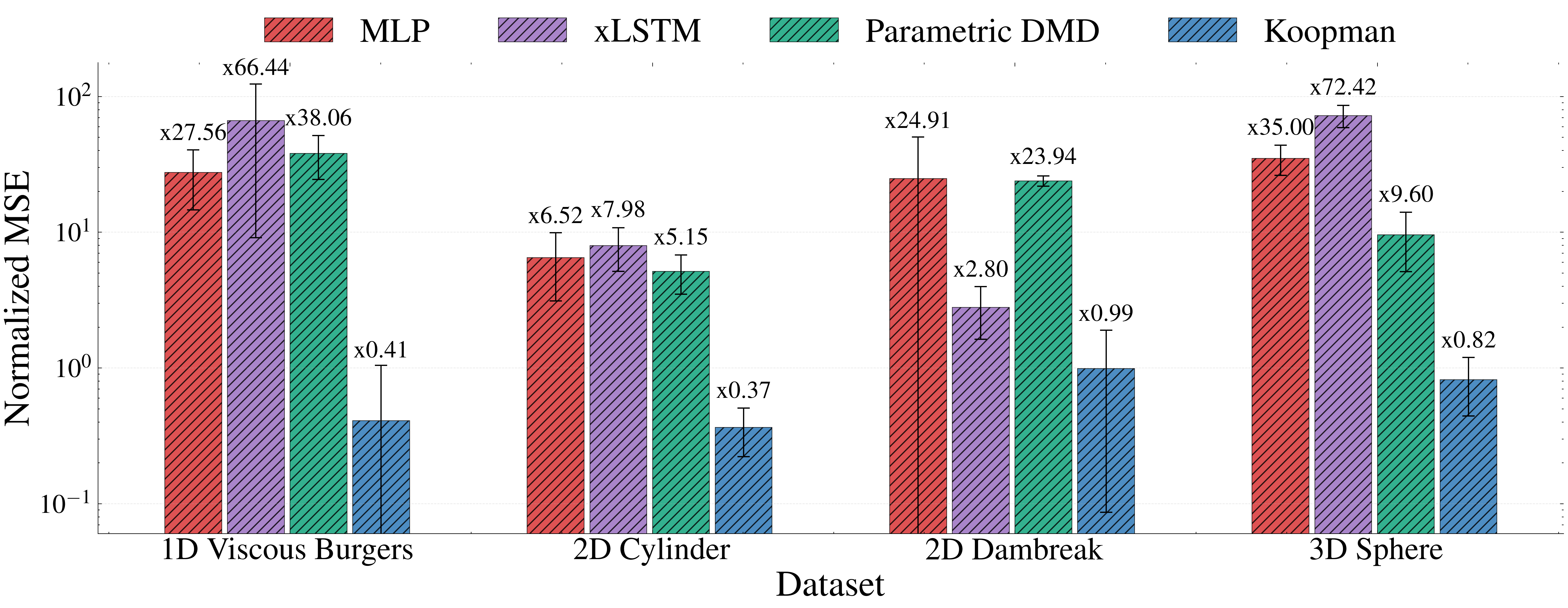}
        \end{minipage}
    \end{subfigure}

    \caption{Comparison of the different latent models in forward-run and DA settings using the exact initial condition, latent dimension $r=32$, and measurement noise level $20\%$. (a) Forward run. (b) Data assimilation. The same task-specific conclusion holds: Koopman dominates the assimilation task despite not having the best rollout performance. }
    \label{fig:main_exact}
\end{figure}

\begin{figure}[!htbp]
    \centering
    \includegraphics[width=0.92\linewidth]{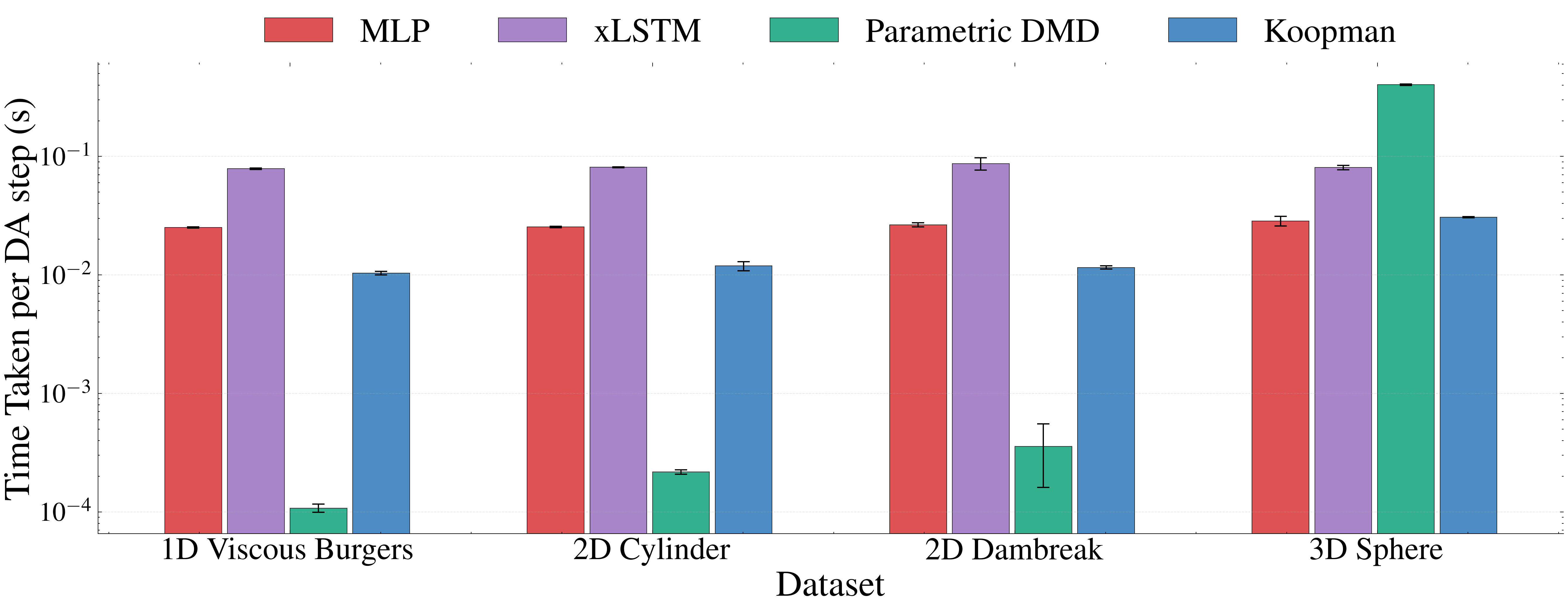}
    \caption{Comparison of the average wall-clock time per DA step for the exact-initialization setting, with latent dimension $r=32$ and measurement noise level $20\%$. }
    \label{fig:main_exact_time}
\end{figure}


The comparison between Koopman and parametric DMD is particularly revealing. Both baselines rely on linear time propagation in a reduced representation, yet their DA performance differs dramatically. The key distinction is that the Koopman model provides an explicit parameter-conditioned operator obtained from a continuous-time generator with skew-symmetric and dissipative components, so the discrete-time map $\mathbf{K}(\mu;\Delta t)$ remains spectrally controlled. The monolithic pDMD baseline instead combines a global POD basis, a global DMD forecast in stacked coefficient space, and interpolation across parameter space, without an explicit parameter-conditioned generator. This indicates that even among reduced models with linear time evolution, the particular operator parameterization produces substantial differences in assimilation performance; we do not attempt to apportion the effect between operator parameterization and the linear--nonlinear distinction, which this comparison does not separate. We quantify this in \cref{subsec:idealized_filtering_perspective}.

\subsection{Temporal error evolution and qualitative reconstruction}\label{subsec:temporal_results}

The final-time bar charts summarize the task, but the trajectory-wise error histories in \Cref{fig:temporal_1d_vb,fig:temporal_2d_cylinder,fig:temporal_2d_dambreak,fig:temporal_3d_sphere} show why the ranking changes under assimilation. In the forward-run setting, the nonlinear baselines can remain accurate for part of the horizon and then drift. When assimilation is activated, their per-trajectory error histories become more similar to one another but settle at a larger error than the corresponding forward runs, so the spread across test cases narrows while the error itself does not fall. These panels show error histories for different test trajectories rather than uncertainty bands around a single estimate, so the narrowing is a statement about consistency across cases and not about the filter's own uncertainty. The Koopman model behaves differently. Its error histories remain low and stable under assimilation, and the filtered trajectories do not show the same systematic upward shift.

This difference is visible in the representative qualitative reconstruction in \Cref{fig:qualitative_1d_vb,fig:cylinder_comp_main,fig:qualitative_2d_dambreak,fig:qualitative_3d_sphere}. In the 1D Burgers example, the Koopman reconstruction tracks the sharp transition and phase of the solution more closely than the competing models. For the 2D cylinder example, the Koopman estimate preserves the phase and amplitude of the wake structure across the trajectory, whereas the MLP and xLSTM either over-damp the wake or introduce distorted structures. In the 2D dambreak case, Koopman and xLSTM both recover the main interface location, but Koopman remains cleaner and more consistent over time. For the 3D sphere streamline slice, Koopman is uniformly the closest reconstruction to the ground truth.

\begin{figure}[!htbp]
    \centering
    \includegraphics[width=0.95\linewidth]{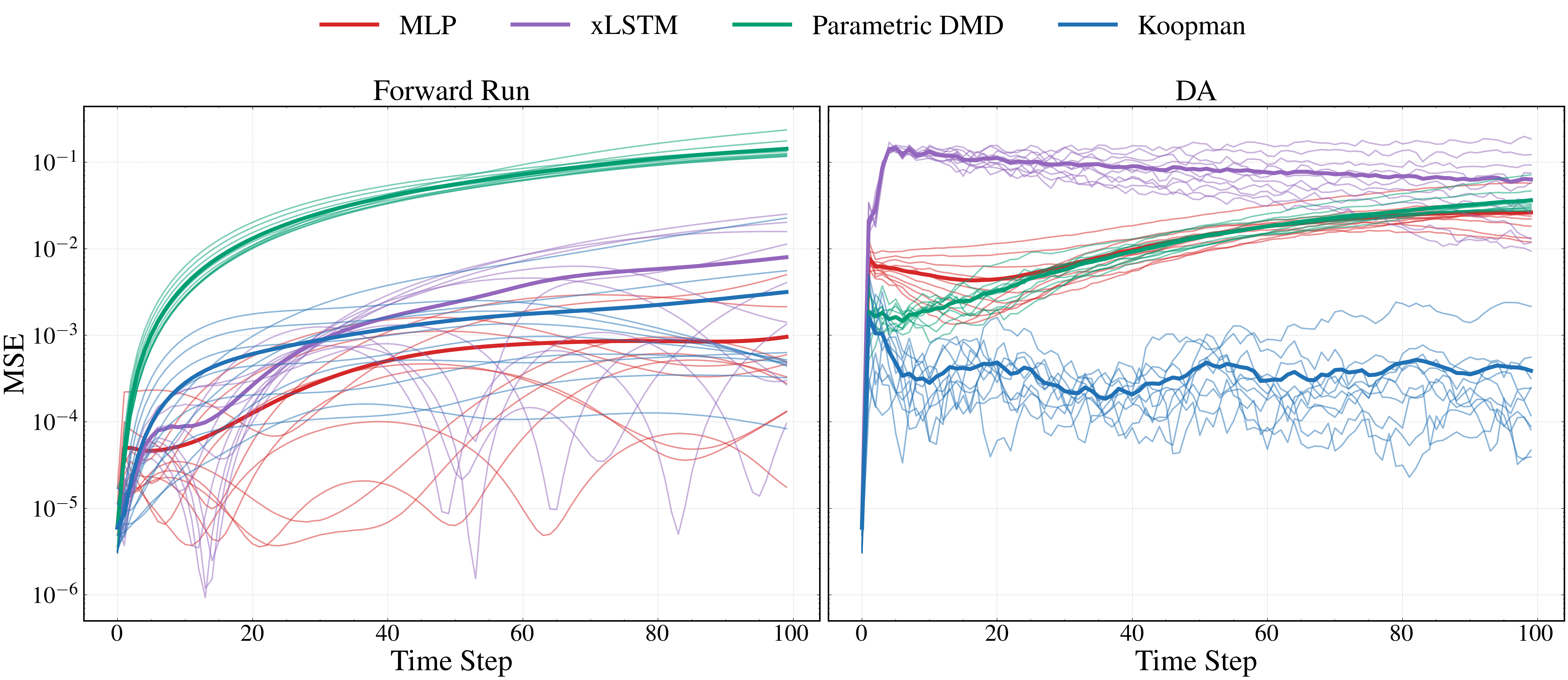}
    \caption{Temporal evolution of the forecast and assimilation errors for approximate initialization on 1D Viscous Burgers. Transparent lines show the error progression with time for each testing trajectory and the solid line represents the mean across all testing trajectories for a given forward model.}
    \label{fig:temporal_1d_vb}
\end{figure}

\begin{figure}[!htbp]
    \centering
    \includegraphics[width=0.95\linewidth]{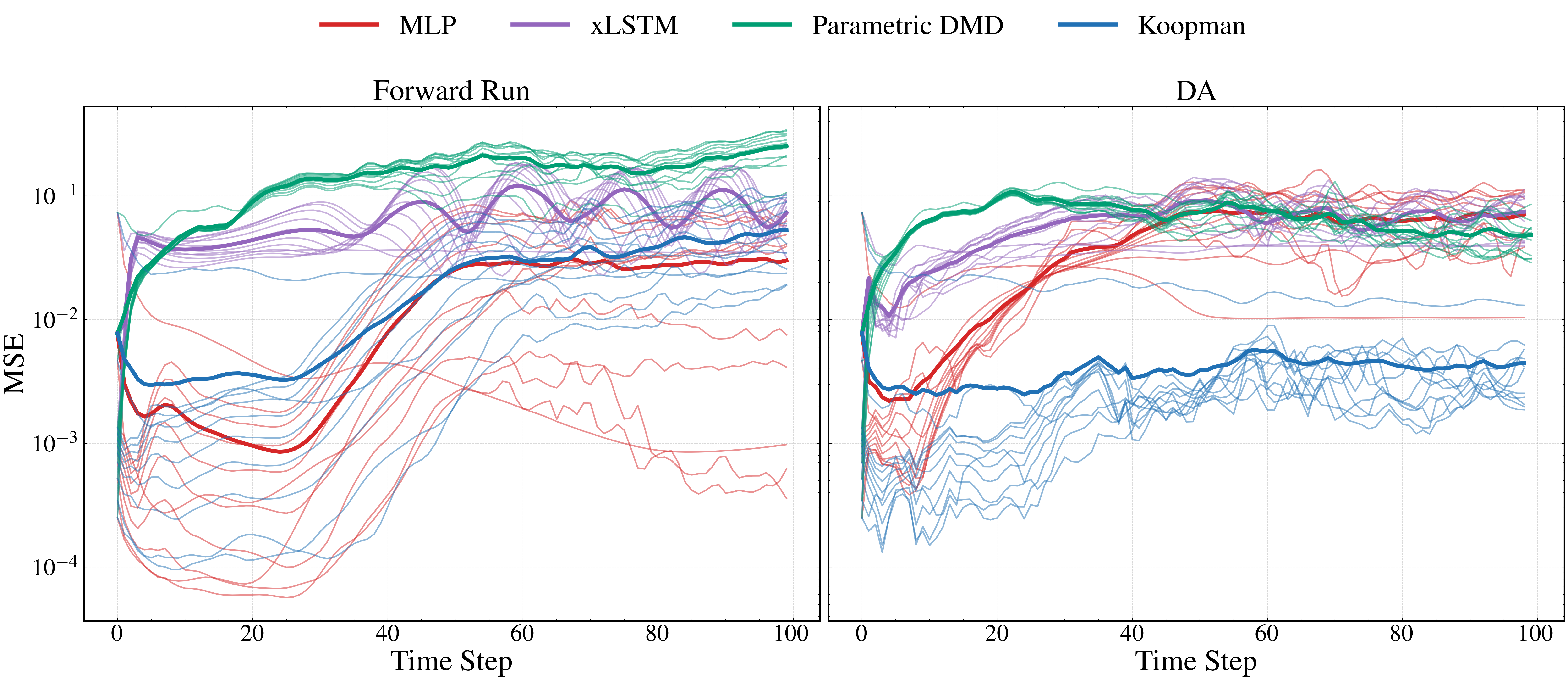}
    \caption{Temporal evolution of the forecast and assimilation errors for approximate initialization on the 2D cylinder. Transparent lines show the error progression with time for each testing trajectory and the solid line represents the mean across all testing trajectories for a given forward model.}
    \label{fig:temporal_2d_cylinder}
\end{figure}

\begin{figure}[!htbp]
    \centering
    \includegraphics[width=0.95\linewidth]{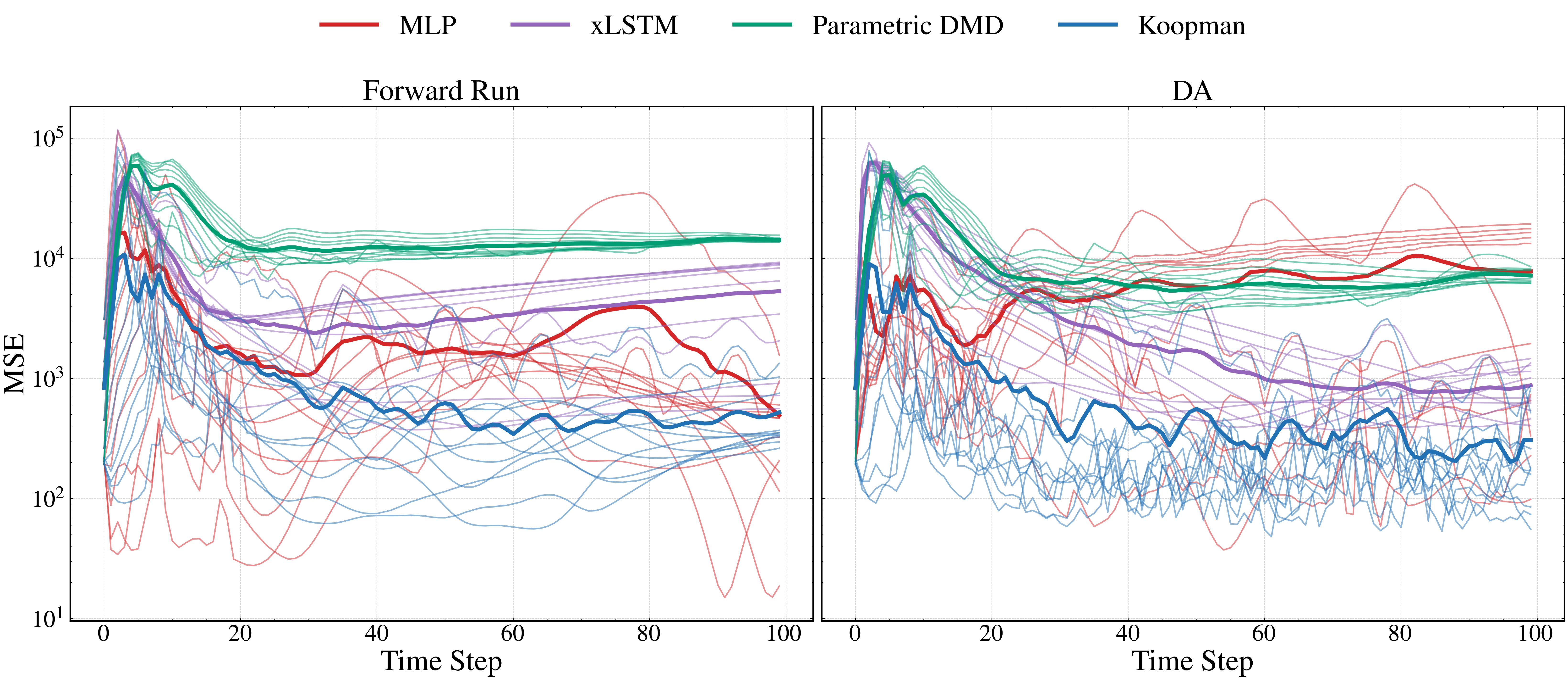}
    \caption{Temporal evolution of the forecast and assimilation errors for approximate initialization on 2D dambreak. Transparent lines show the error progression with time for each testing trajectory and the solid line represents the mean across all testing trajectories for a given forward model.}
    \label{fig:temporal_2d_dambreak}
\end{figure}

\begin{figure}[!htbp]
    \centering
    \includegraphics[width=0.95\linewidth]{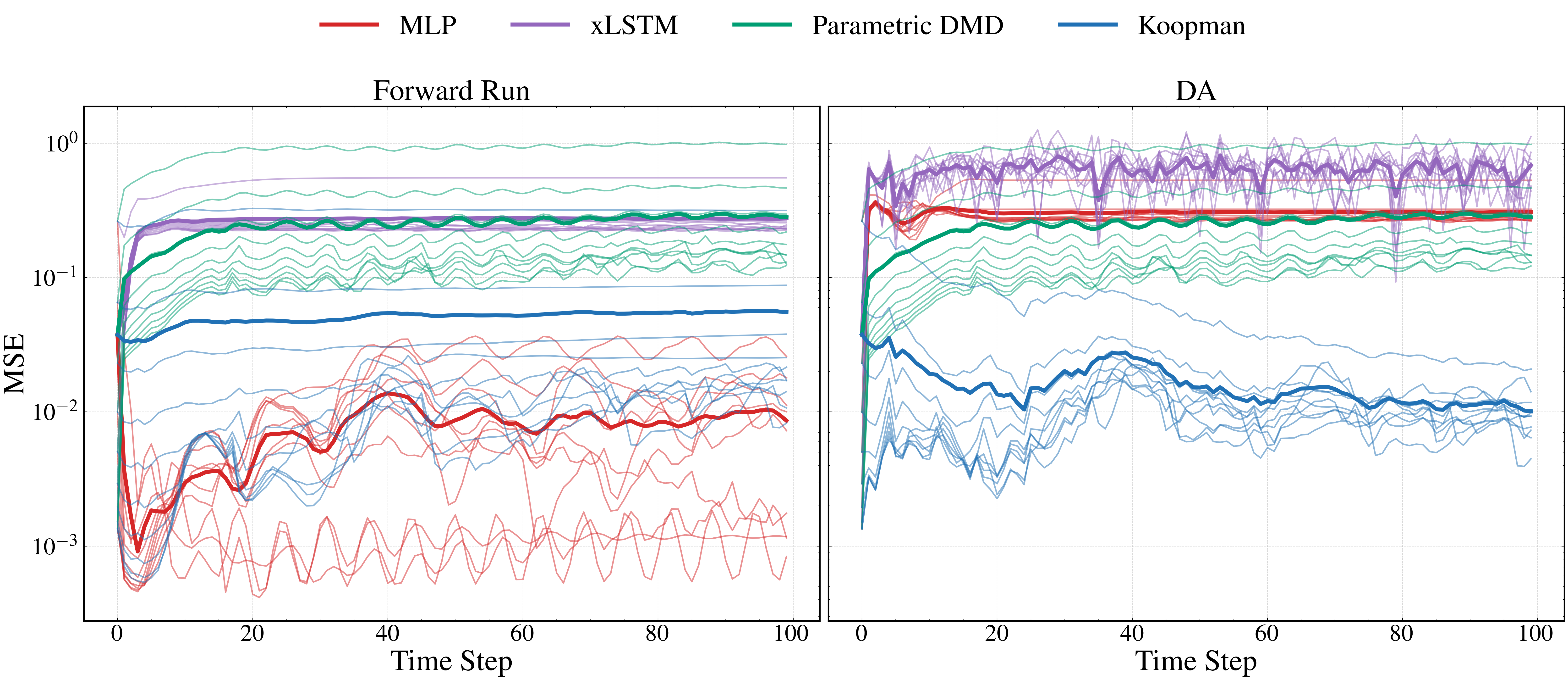}
    \caption{Temporal evolution of the forecast and assimilation errors for approximate initialization on the 3D sphere. Transparent lines show the error progression with time for each testing trajectory and the solid line represents the mean across all testing trajectories for a given forward model.}
    \label{fig:temporal_3d_sphere}
\end{figure}

\begin{figure}[!htbp]
    \centering
    \includegraphics[width=0.90\linewidth]{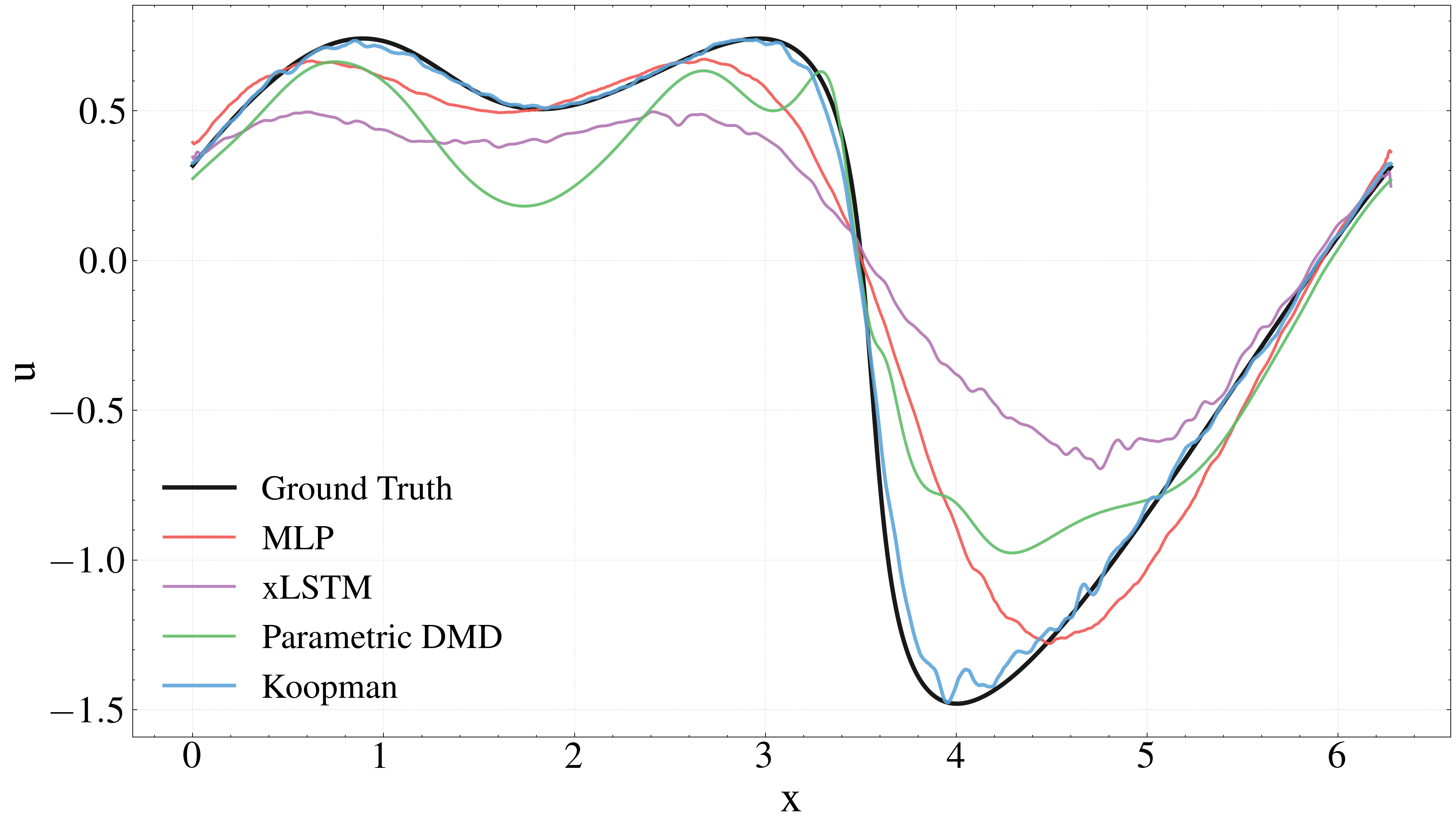}
    \caption{State values at the final timestep for 1D Viscous Burgers equation for a testing trajectory with $\nu=0.01$ under approximate initialization, latent dimension $r=32$, and $20\%$ measurement noise.}
    \label{fig:qualitative_1d_vb}
\end{figure}

\begin{figure}[!htbp]
    \centering
    \includegraphics[width=\linewidth]{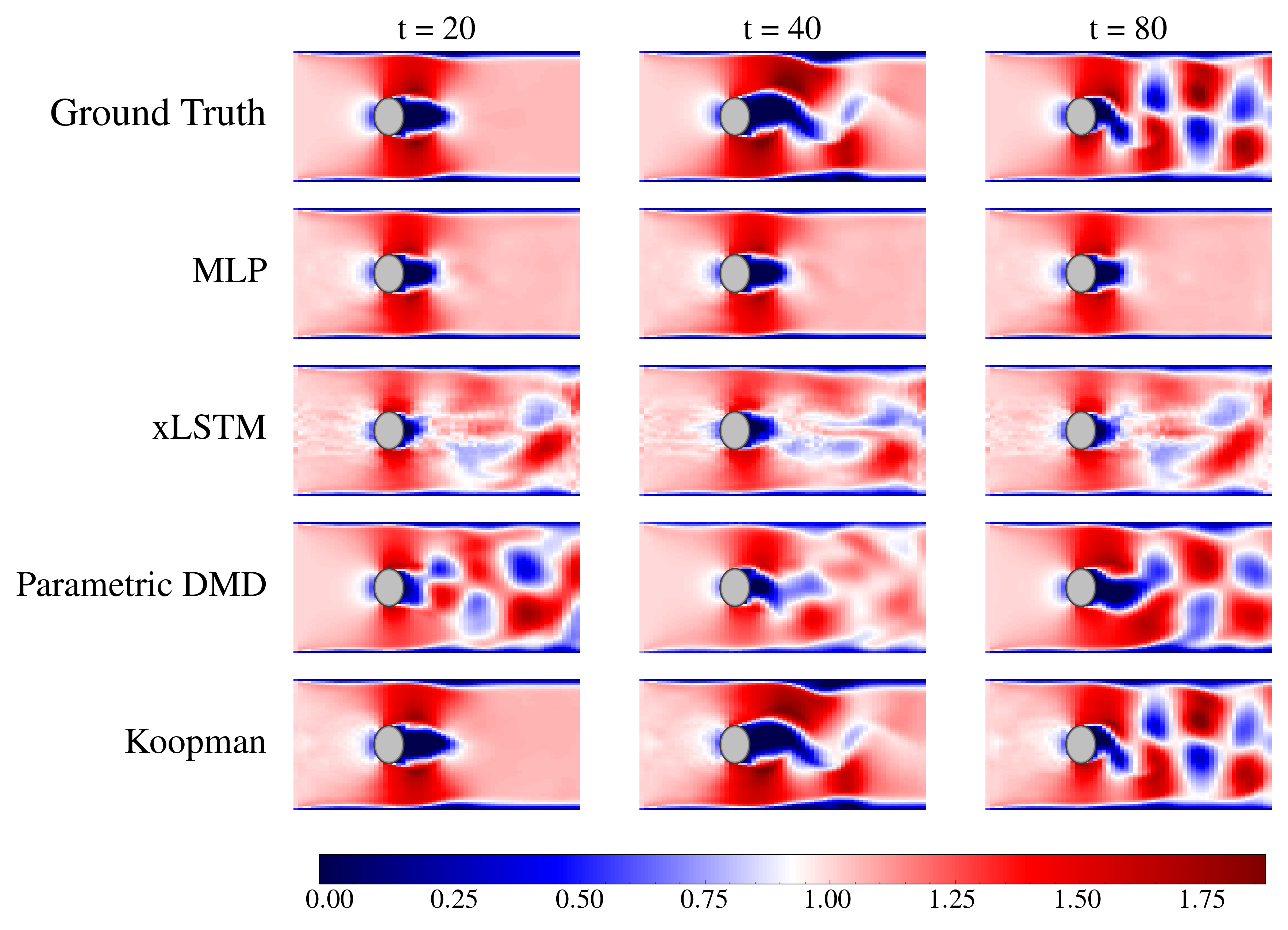}
    \caption{Reconstruction of 2D Cylinder velocity magnitude for a testing trajectory of $Re=715$ under data assimilation with approximate initialization, latent dimension $r=32$, and $20\%$ of measurement noise.}
    \label{fig:cylinder_comp_main}
\end{figure}

\begin{figure}[!htbp]
    \centering
    \includegraphics[width=\linewidth]{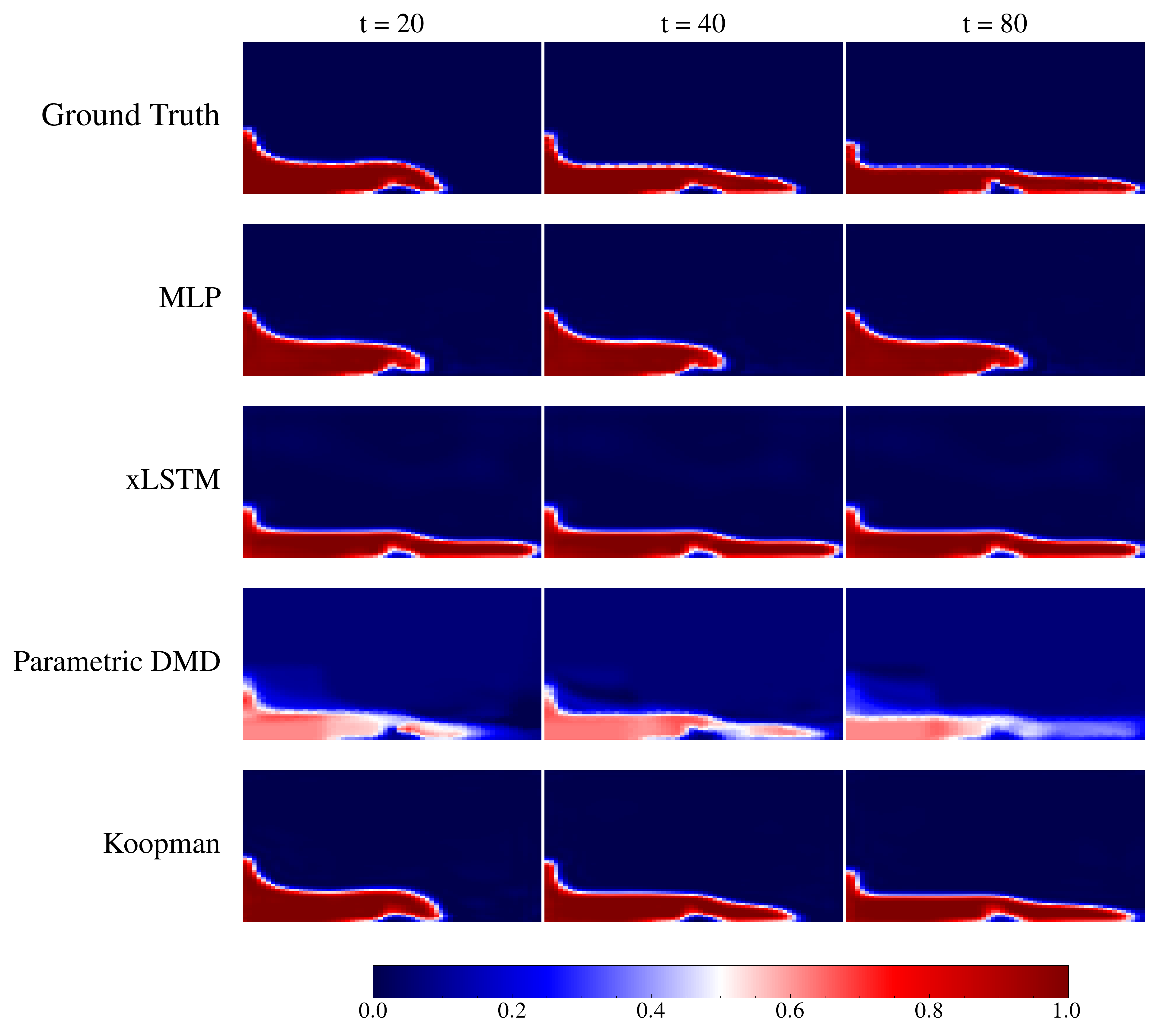}
    \caption{Reconstruction of 2D Dambreak liquid volume fraction for a testing trajectory of $\nu_{w}=0.08\,\mathrm{m}^{2}\mathrm{s}^{-1}$ under data assimilation with approximate initialization, latent dimension $r=32$, and $20\%$ of measurement noise.}
    \label{fig:qualitative_2d_dambreak}
\end{figure}

\begin{figure}[!htbp]
    \centering
    \includegraphics[width=\linewidth]{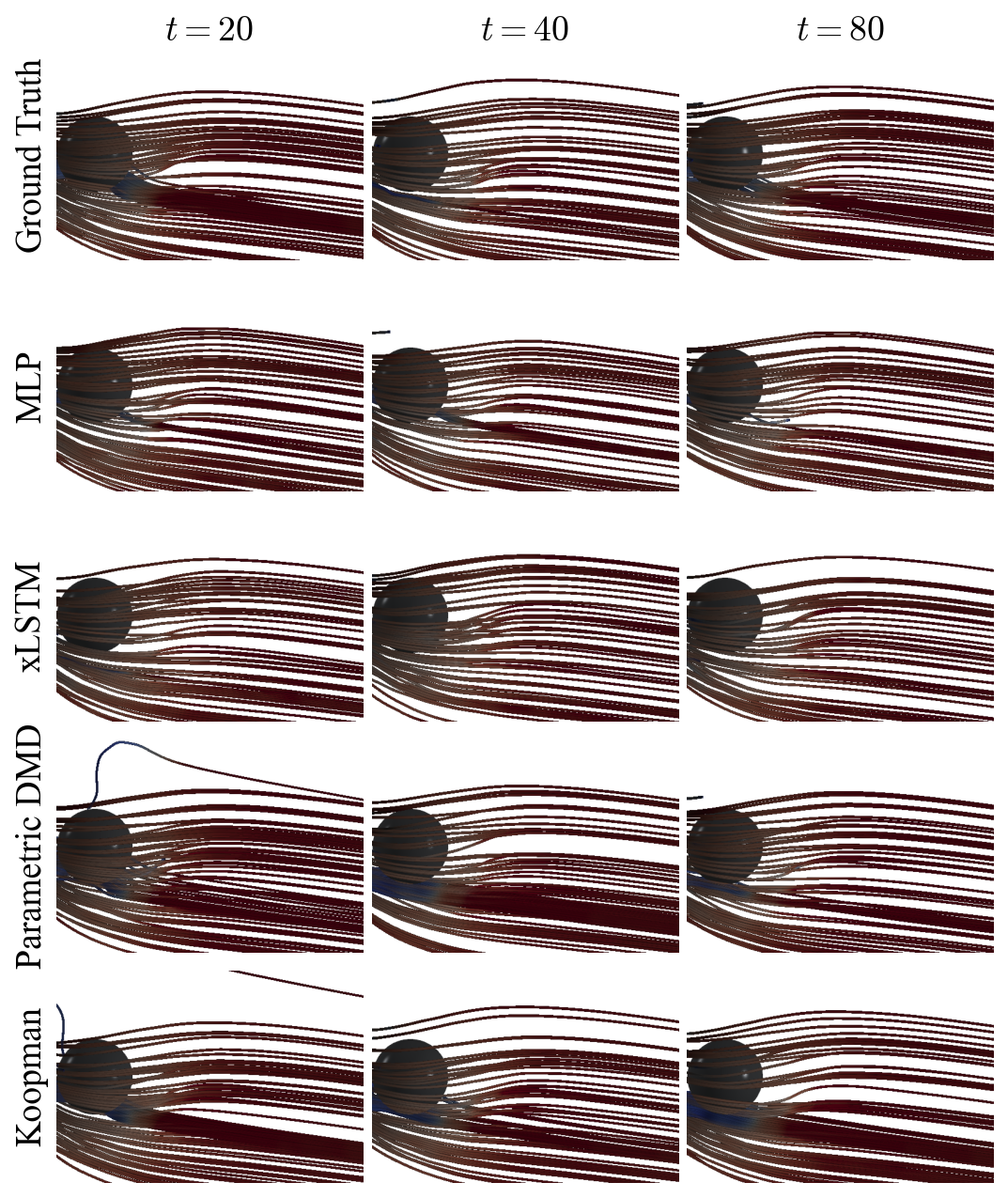}
    \caption{Reconstruction of streamlines (colored by velocity magnitude) over 3D Sphere for a testing trajectory of $Re=202$ under data assimilation with approximate initialization, latent dimension $r=32$, and $20\%$ of measurement noise.}
    \label{fig:qualitative_3d_sphere}
\end{figure}
\subsection{Sensitivity to measurement noise and latent dimension}\label{app:ablation}

We study the sensitivity of sensing performance to measurement noise. We present the DA results here under approximate initialization. The corresponding forward-run comparison is already shown in panel (a) of \Cref{fig:main_inexact}.

\paragraph*{Measurement-noise sensitivity.}
\Cref{fig:noise_ablation} varies the observation noise while keeping the latent dimension fixed at $r=32$. Koopman remains the strongest DA model on all four datasets at every tested observation-noise level.

\begin{figure}[!htbp]
    \centering
    \begin{subfigure}[t]{0.92\linewidth}
        \centering
        \includegraphics[width=\linewidth]{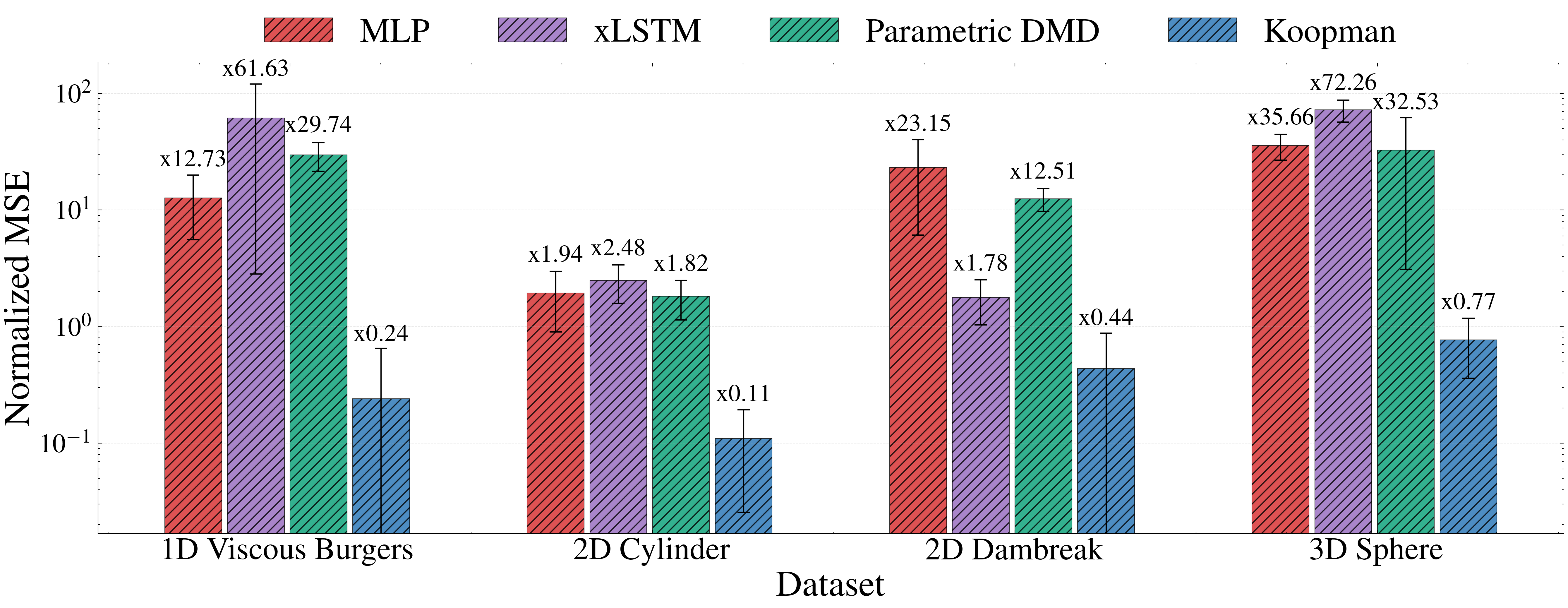}
        \caption{DA with $10\%$ noise.}
    \end{subfigure}

    \vspace{0.5em}

    \begin{subfigure}[t]{0.92\linewidth}
        \centering
        \includegraphics[width=\linewidth]{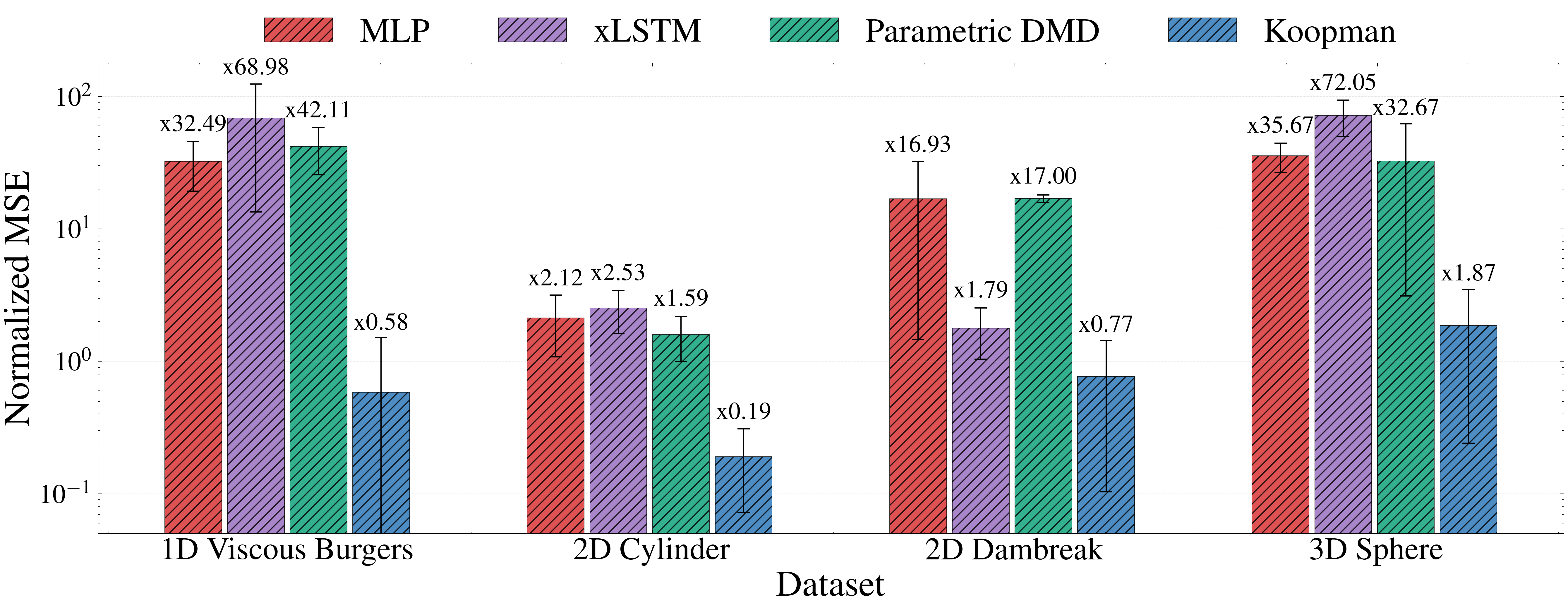}
        \caption{DA with $30\%$ noise.}
    \end{subfigure}

    \caption{Measurement-noise ablation for approximate initialization and latent dimension $r=32$. Across both noise levels, Koopman remains the strongest DA model. }
    \label{fig:noise_ablation}
\end{figure}

\paragraph*{Latent-dimension sensitivity.}
\Cref{fig:latent_ablation} compares smaller and larger latent spaces against the $r=32$. At $r=16$, Koopman already gives the lowest DA error on all four datasets. At $r=64$, the gap widens on some problems, especially 1D Burgers and the 3D sphere, where the nonlinear baselines degrade more noticeably under DA (on the sphere the MLP-to-Koopman DA error ratio grows from about $30$ at $r=32$ to about $55$ at $r=64$). The conclusion is therefore not tied to a single latent dimension.

\begin{figure}[!htbp]
    \centering
    \begin{subfigure}[t]{0.92\linewidth}
        \centering
        \includegraphics[width=\linewidth]{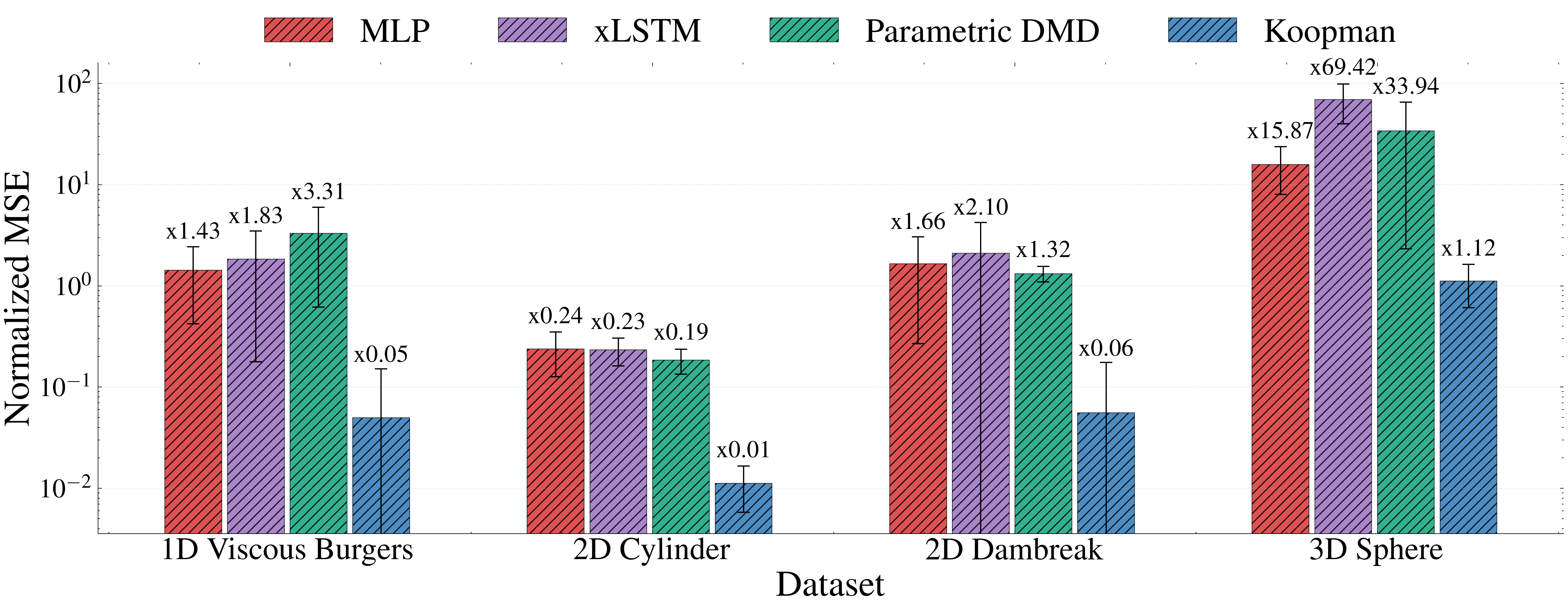}
        \caption{DA with $r=16$.}
    \end{subfigure}

    \vspace{0.5em}

    \begin{subfigure}[t]{0.92\linewidth}
        \centering
        \includegraphics[width=\linewidth]{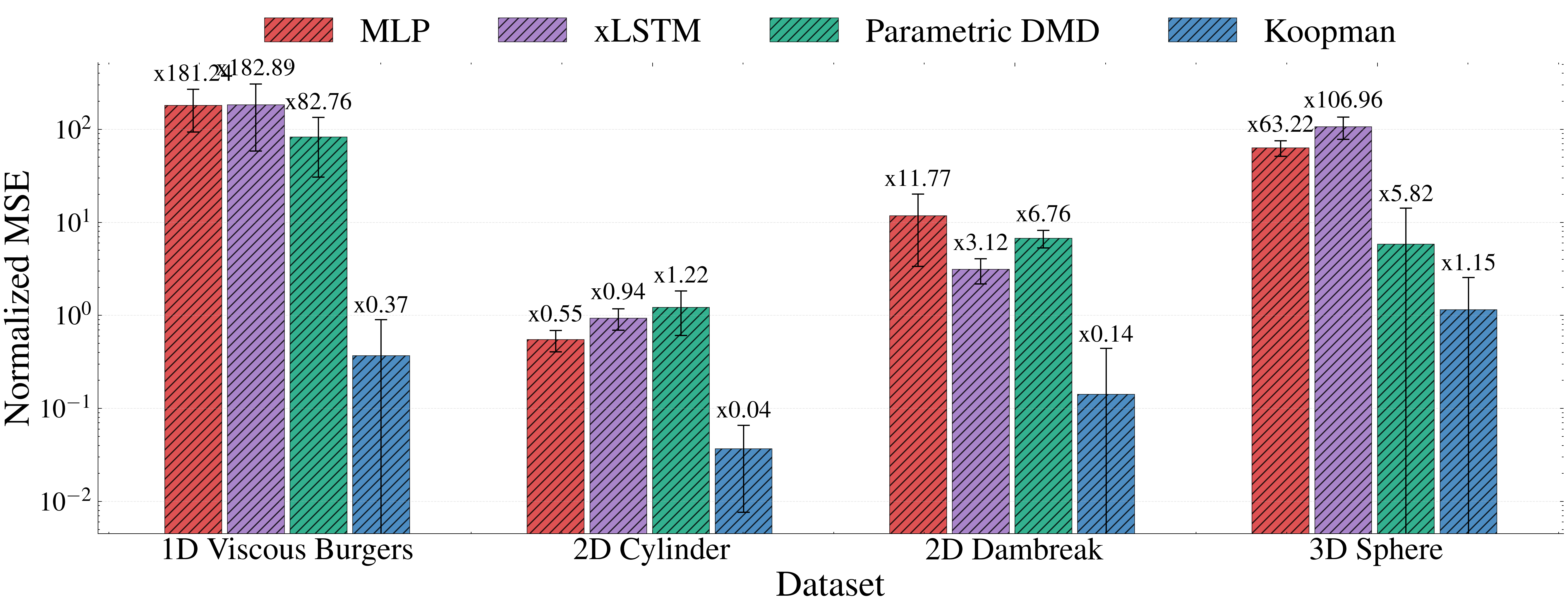}
        \caption{DA with $r=64$.}
    \end{subfigure}

    \caption{Latent-dimension ablation for approximate initialization and $20\%$ measurement noise. Relative to the $r=32$ main-text setting, Koopman remains the strongest DA model at both smaller and larger latent dimensions. Each panel is normalized by the forward-run MLP error at the latent dimension of that panel, so values are comparable within a panel but not across panels.}
    \label{fig:latent_ablation}
\end{figure}

\subsection{Ensemble Kalman Filter}\label{sec:enkf}

To test whether the Koopman advantage is tied specifically to Jacobian-based filtering, we repeat the inexact-initialization experiment with an Ensemble Kalman Filter (EnKF). We use 20 ensemble members and initialize the filter with the same numerical covariance settings as the EKF, so the two filters are compared under the same nominal configuration. These values are maintained across all models and experiments. The qualitative conclusion remains largely unchanged. As shown in \cref{fig:enkf_results}, Koopman attains the lowest EnKF error on 1D Burgers, the 2D cylinder and the 2D dambreak; only on the 3D sphere is it beaten, by the xLSTM ($6.5\times10^{-3}$ against $1.2\times10^{-2}$). The Burgers margin over the MLP narrows sharply relative to the EKF, from a factor of about $67$ to about $1.2$. The change is systematic and it is worth stating precisely, because it is the sharpest
evidence in the paper for the mechanism proposed in \cref{subsec:idealized_filtering_perspective}.
\Cref{tab:enkf_accuracy} compares each model's open-loop rollout with its EnKF estimate under the
same sensors, noise and initialization.

\begin{table}[!htbp]
\centering
\small
\setlength{\tabcolsep}{6pt}
\renewcommand{\arraystretch}{1.15}
\begin{adjustbox}{max width=\linewidth}
\begin{tabular}{llcccc}
\toprule
Model & Estimate & 1D Burgers & 2D cylinder & 2D dambreak & 3D sphere \\
\midrule
Koopman & rollout & $3.17\!\times\!10^{-3}$ & $5.33\!\times\!10^{-2}$ & $5.23\!\times\!10^{2}$ & $5.56\!\times\!10^{-2}$ \\
 & EnKF & $7.62\!\times\!10^{-4}$ & $8.51\!\times\!10^{-3}$ & $3.80\!\times\!10^{2}$ & $1.23\!\times\!10^{-2}$ \\
pDMD & rollout & $1.43\!\times\!10^{-1}$ & $2.54\!\times\!10^{-1}$ & $1.43\!\times\!10^{4}$ & $2.81\!\times\!10^{-1}$ \\
 & EnKF & $6.22\!\times\!10^{-2}$ & $5.86\!\times\!10^{-2}$ & $7.74\!\times\!10^{3}$ & $1.15\!\times\!10^{-1}$ \\
MLP & rollout & $9.60\!\times\!10^{-4}$ & $3.01\!\times\!10^{-2}$ & $4.91\!\times\!10^{2}$ & $8.60\!\times\!10^{-3}$ \\
 & EnKF & $9.27\!\times\!10^{-4}$ & $2.31\!\times\!10^{-2}$ & $5.13\!\times\!10^{2}$ & $9.13\!\times\!10^{-3}$ \\
xLSTM & rollout & $8.00\!\times\!10^{-3}$ & $7.31\!\times\!10^{-2}$ & $5.34\!\times\!10^{3}$ & $2.76\!\times\!10^{-1}$ \\
 & EnKF & $1.42\!\times\!10^{-3}$ & $7.31\!\times\!10^{-2}$ & $5.34\!\times\!10^{3}$ & $\mathbf{6.45\!\times\!10^{-3}}$ \\
\midrule
Koopman & EKF & $\mathbf{3.90\!\times\!10^{-4}}$ & $\mathbf{4.45\!\times\!10^{-3}}$ & $\mathbf{3.06\!\times\!10^{2}}$ & $1.01\!\times\!10^{-2}$ \\
\bottomrule
\end{tabular}
\end{adjustbox}
\caption{Final-time MSE under the ensemble Kalman filter with $n=20$ members, compared with each
model's own open-loop rollout, for the main experimental setting ($r=32$, $m=16$ sensors, $20\%$
measurement noise, approximate initialization). Under the EnKF the MLP improves on two of the four
benchmarks instead of none and the xLSTM improves or is inert on all four, so the
linear--nonlinear split of \cref{sec:results} is specific to the EKF. The last column repeats the
Koopman EKF result of the main experiment for reference.}
\label{tab:enkf_accuracy}
\end{table}

Under the EnKF the MLP improves on two of the four benchmarks instead of none, and the xLSTM
improves or is inert on all four. The linear--nonlinear split of \cref{sec:results} is
therefore specific to the EKF and is not a property of nonlinear latent dynamics as such. This is
consistent with the decomposition of \cref{lem:local_ekf_error_decomposition} in
\cref{subsec:idealized_filtering_perspective}, in which the forecast-linearization residual
$\boldsymbol{\rho}^f$ enters only through the explicit Jacobian $\mathbf{F}_k$ and vanishes for a
linear transition; an ensemble filter forms no such Jacobian. 
\Cref{lem:local_ekf_error_decomposition} is an EKF error identity whose remainder is a signed vector
rather than a one-sided penalty, and an ensemble filter is not the same filter with one term deleted:
it propagates members through the nonlinear forecast and estimates the moments from that ensemble.
The comparison shows the nonlinear baselines to be sensitive to the EKF's local approximation, but it
does not isolate forecast linearization as the sole cause.

The last column of \cref{tab:enkf_accuracy} shows that Koopman under the
EKF is the most accurate configuration on 1D Burgers, the 2D cylinder and the 2D
dambreak, beating \emph{every} EnKF result including its own, and it is beaten only on the 3D sphere.
The margin over the best EnKF configuration is a factor of $2.0$ on Burgers
($3.90\times10^{-4}$ against $7.62\times10^{-4}$), $1.9$ on the cylinder ($4.45\times10^{-3}$ against
$8.51\times10^{-3}$) and $1.2$ on the dambreak ($3.06\times10^{2}$ against $3.80\times10^{2}$).


The cost of the ensemble filter lives in the ensemble size $n$, which multiplies the number of
forward passes at every step. \Cref{tab:ensemble_size} shows the Koopman EnKF error falling by a factor of $1.9$ between $n=20$ and $n=40$ and then changing little, to the level of the Koopman EKF, and the MLP improving by $34\%$ over the range, almost all of it by $n=40$. The xLSTM behaves differently: its error is the same at every ensemble size and equals its own open-loop error to three significant figures. Its forecast ensemble contracts strongly. Averaged over the eleven cylinder test trajectories, the ensemble spread falls from about $4.6$ at the first analysis step to $0.12$--$0.13$ within ten steps and stays there, and the spectral norm of the Kalman gain falls from about $4$ to about $10^{-2}$, against $0.15$--$1.4$ for the Koopman model and $6$--$8$ for the MLP over the same steps. The measurement updates therefore change the xLSTM estimate very little, and its final-time error is indistinguishable from its open-loop error at the reported precision. This does not show that the filtered and open-loop trajectories coincide, since updates made while the gain is still large can move the estimate. Enlarging the ensemble to $n=64$ leaves both the gain and the error essentially unchanged, which is consistent with the contraction coming from the learned latent map rather than from the sample size. Each entry is a single run, so these are trends rather than converged values. Recovering the
MLP under an ensemble filter requires an ensemble large enough to resolve the latent
covariance, and the cost scales linearly with it. Under the EKF, Koopman reaches the same cylinder accuracy with no ensemble at all.
The main advantage of Koopman is not simply that its Jacobian is easy to compute; it is that the underlying latent dynamics remain more compatible with sequential state estimation.

\begin{table}[!htbp]
\centering
\setlength{\tabcolsep}{6pt}
\renewcommand{\arraystretch}{1.15}
\begin{adjustbox}{max width=\linewidth}
\begin{tabular}{lccc}
\toprule
EnKF final-time MSE (2D cylinder) & $n=20$ & $n=40$ & $n=64$ \\
\midrule
Koopman & $8.51\!\times\!10^{-3}$ & $4.46\!\times\!10^{-3}$ & $4.42\!\times\!10^{-3}$ \\
MLP     & $2.31\!\times\!10^{-2}$ & $1.56\!\times\!10^{-2}$ & $1.52\!\times\!10^{-2}$ \\
xLSTM   & $7.31\!\times\!10^{-2}$ & $7.31\!\times\!10^{-2}$ & $7.31\!\times\!10^{-2}$ \\
\bottomrule
\end{tabular}
\end{adjustbox}
\caption{Effect of the ensemble size on the 2D cylinder, with everything but $n$ held fixed at the
main experimental setting, computed with the same filter implementation as \cref{tab:enkf_accuracy}, whose $n=20$ column it reproduces. Each additional member is an additional forward pass of the latent model at every
step, so accuracy that has to be bought with a larger ensemble is bought at a proportionate cost.}
\label{tab:ensemble_size}
\end{table}

\begin{figure}[!htbp]
    \centering
    \includegraphics[width=0.92\linewidth]{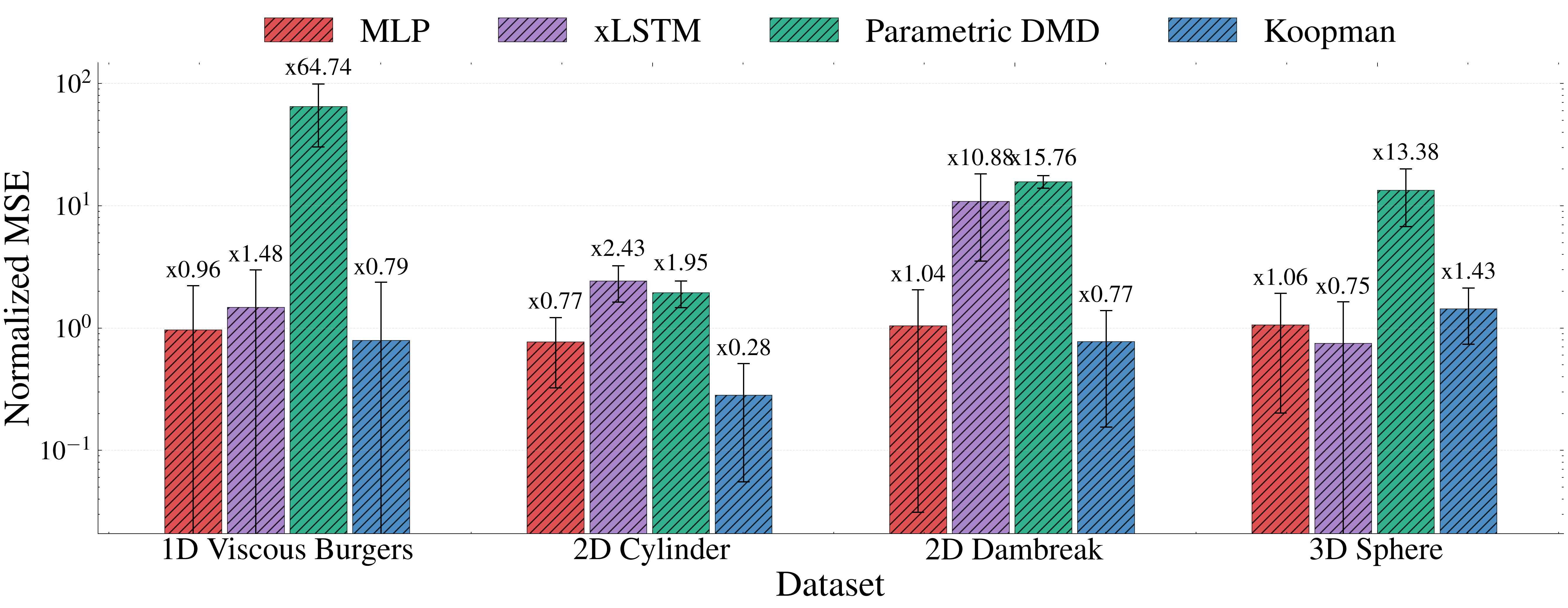}
    \caption{EnKF assimilation results for approximate initialization, latent dimension $r=32$, and $20\%$ measurement noise. Koopman attains the lowest EnKF error on the 2D cylinder and 2D dambreak and is close to the best on 1D Burgers. Unlike under the EKF, the nonlinear models are no longer systematically degraded by assimilation.}
    \label{fig:enkf_results}
\end{figure}



\subsection{Performance under multi-step training}

The main results are reported for one-step training. To assess whether the observed advantage of Koopman latent dynamics is an artifact of short-horizon optimization, we additionally train MLP and Koopman baselines (since they are the two most competitive models in the study) with multi-step objectives using 2-, 4-, and 8-step losses. Each panel of \cref{fig:multistep_comparison_main} is normalized by the MLP rollout error at
that same training horizon, so the panels compare models within a horizon and are not on a common
scale across horizons; the figure therefore supports comparisons between models at a fixed horizon
rather than statements about how the absolute error moves as the horizon grows. In absolute terms the
rollout error does move with the horizon, and not monotonically: the MLP cylinder rollout error is
$3.01\times10^{-2}$, $1.42\times10^{-1}$, $1.28\times10^{-2}$ and $2.12\times10^{-2}$ at horizons
$1$, $2$, $4$ and $8$. The assimilation ordering is nonetheless unchanged at every horizon; on the
cylinder the Koopman DA error is $4.45\times10^{-3}$, $2.60\times10^{-3}$, $3.14\times10^{-3}$ and
$2.91\times10^{-3}$ against $7.07\times10^{-2}$, $4.53\times10^{-2}$, $6.79\times10^{-2}$ and
$8.24\times10^{-2}$ for the MLP, a margin of between $15$ and $28$ times. Across all training horizons considered, the Koopman model yields consistently lower assimilation error than the MLP baseline. Extending the training horizon therefore does not close the assimilation gap between the models.
\begin{figure}[!htbp]
    \centering

    \begin{subfigure}[!htbp]{\linewidth}
        \centering
        \includegraphics[width=0.70\linewidth]{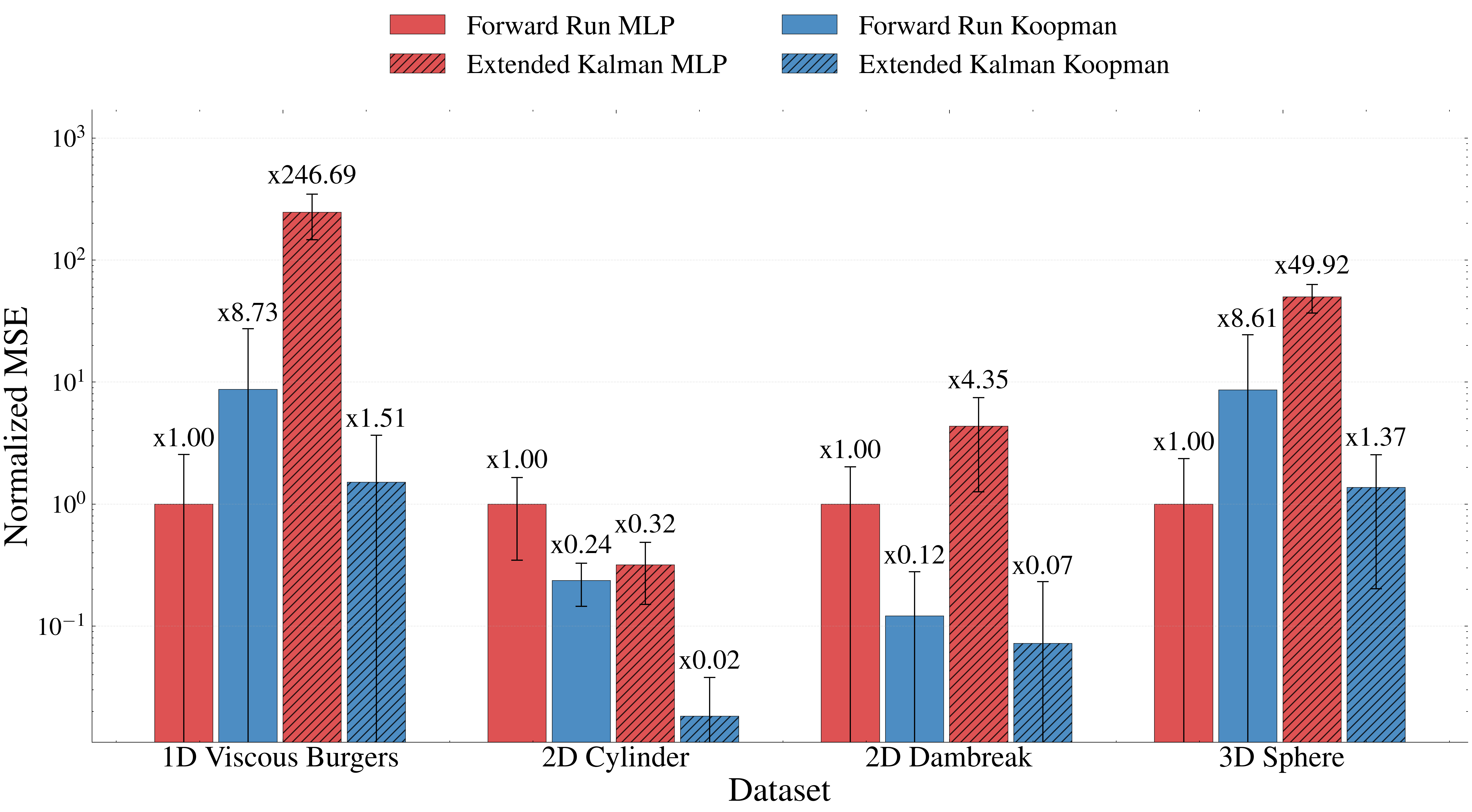}
        \caption{2-step training loss.}
    \end{subfigure}

    \medskip
    \begin{subfigure}[!htbp]{\linewidth}
        \centering
        \includegraphics[width=0.70\linewidth]{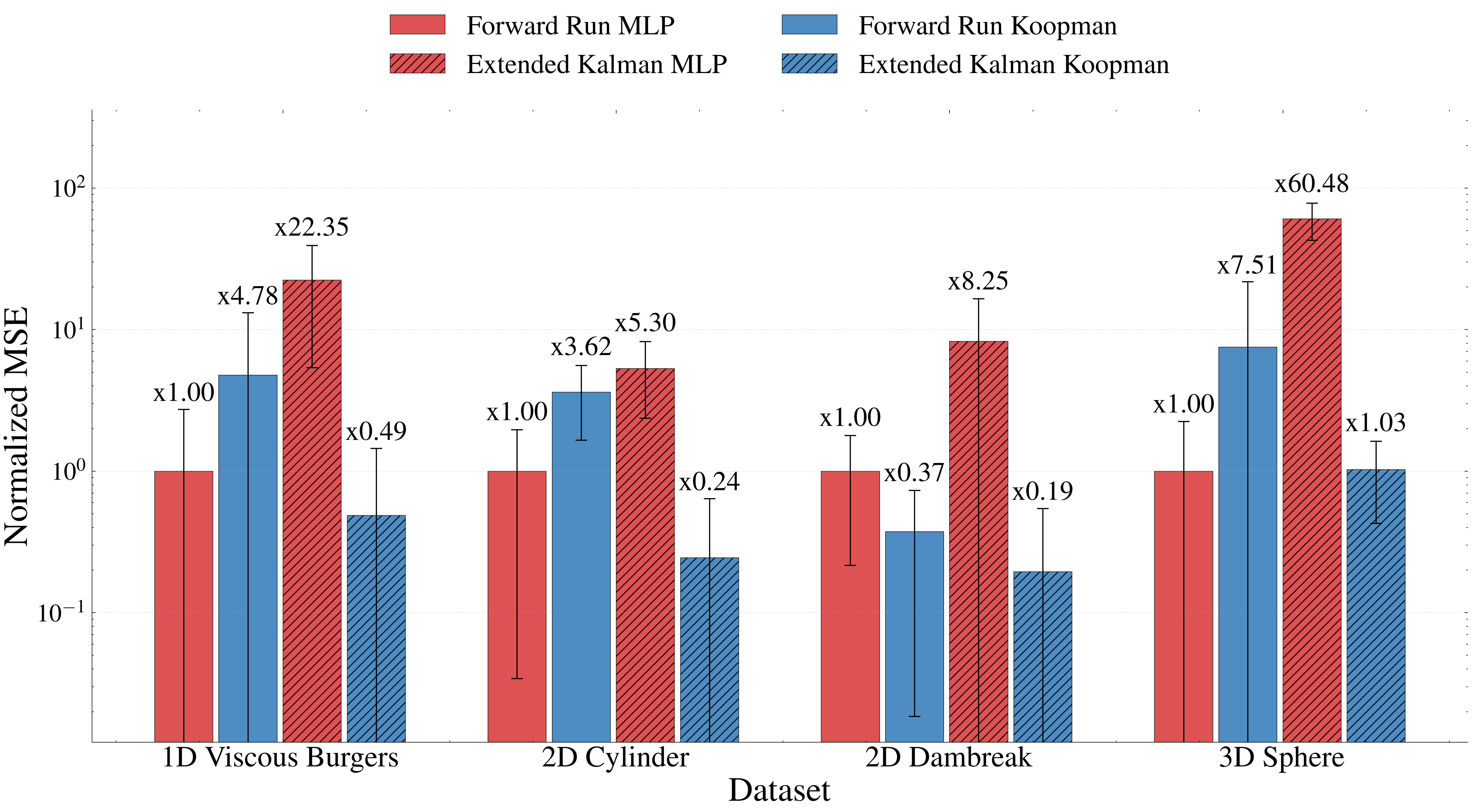}
        \caption{4-step training loss.}
    \end{subfigure}

    \medskip
    \begin{subfigure}[!htbp]{\linewidth}
        \centering
        \includegraphics[width=0.70\linewidth]{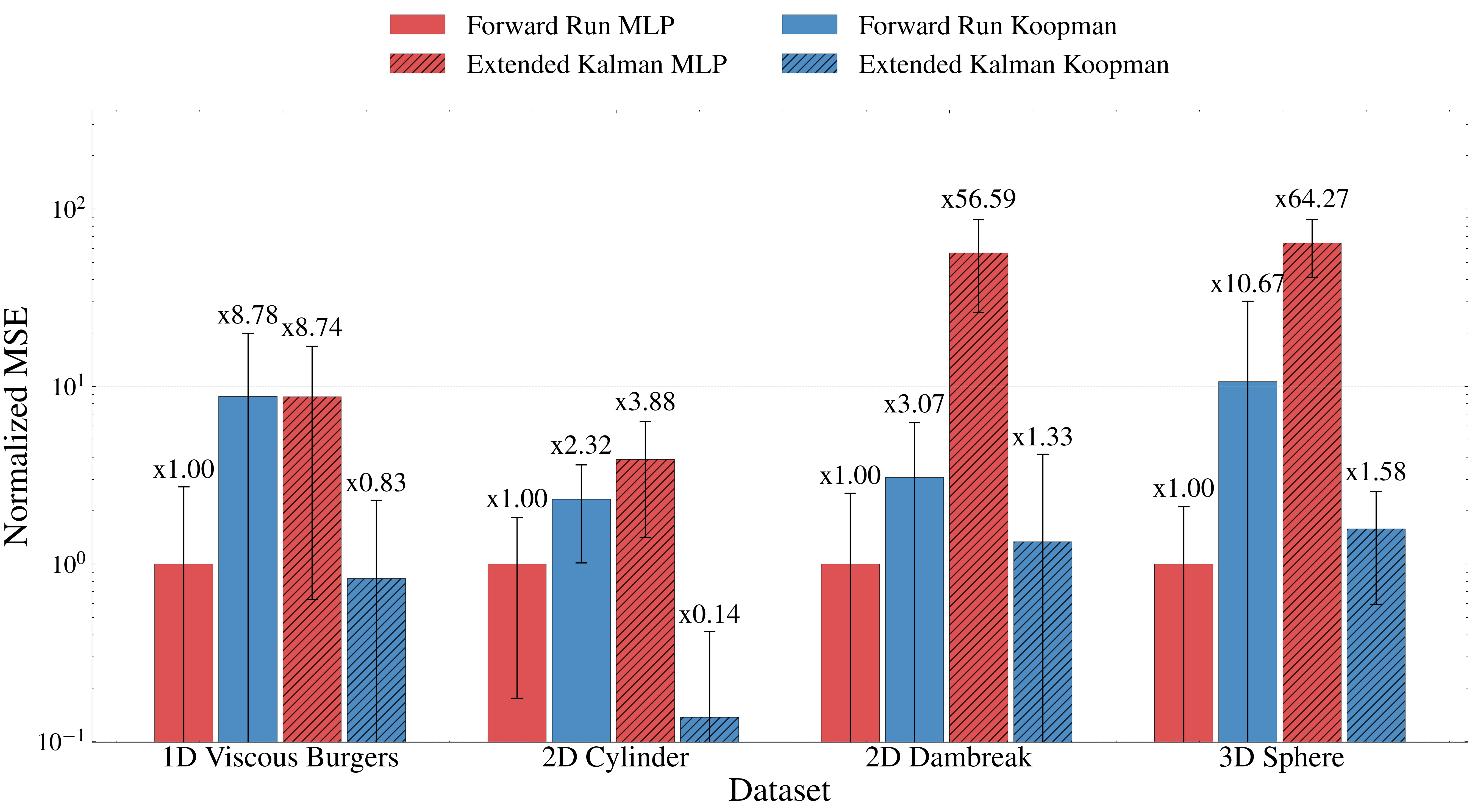}
        \caption{8-step training loss.}
    \end{subfigure}

    \caption{Comparison of MLP and Koopman under multi-step training losses.}
    \label{fig:multistep_comparison_main}
\end{figure}

\subsection{Structural advantages of the proposed latent dynamics for EKF assimilation}
\label{subsec:idealized_filtering_perspective}

The experiments of \cref{sec:results} raise two questions. First, why does the linear Koopman latent model couple with Kalman-type assimilation better than nonlinear latent models built on the same architecture family and parameter budget? Second, why does open-loop rollout accuracy fail to predict assimilation accuracy in the first place? This section states the elementary results that address these questions and discusses what each one explains for the proposed method; the proofs, together with supplementary results, are collected in \cref{app:idealized_filtering_lemmas}. The analysis is local and idealized throughout, and after each result we state what it does not imply. Throughout, $\|\cdot\|$ denotes the Euclidean norm on vectors and the induced spectral norm on matrices, and for a twice continuously differentiable map $f$ with components $f_i$, $\nabla^2 f_i(\mathbf{z})$ denotes the Hessian matrix of the $i$-th component. In particular, $f(\cdot):=\mathcal{F}_{\theta}(\cdot,\mu)$ denotes the learned latent transition at the operating parameter $\mu$, and $h=\mathcal S\circ G_{\psi}$ the decoded sensor map.

\subsubsection{Exact linearization of the learned forecast step}

The EKF advances its covariance through the Jacobian of the latent transition and computes its update through the Jacobian of the decoded sensor map, so the latent model shapes not only the forecast but also the uncertainty model used internally by the filter. Three distinct objects appear in the analysis and should not be conflated: the true field trajectory $\mathbf{x}_k$ of the flow map \cref{eq:flow_map}, whose evolution generates the data; a reference latent trajectory $\mathbf{z}_k$, against which filter errors are measured, with canonical choice the encoded true field, as introduced in \cref{subsec:representation},
\begin{equation}
    \mathbf{z}_k=E_{\phi}(\mathbf{x}_k,\mu);
    \label{eq:reference_latent}
\end{equation}
and the filter's own iterates $\hat{\mathbf{z}}_k^f$ and $\hat{\mathbf{z}}_k^a$. Model-generated quantities carry hats; data-anchored quantities do not. In contrast to a model rollout of \cref{eq:latent_markov}, $\hat{\mathbf{z}}_{k+1}=\mathcal{F}_{\theta}(\hat{\mathbf{z}}_k,\mu)$, which satisfies the learned dynamics exactly by construction, the encoded trajectory does not; its defect
\begin{equation}
    \boldsymbol{\delta}_k=\mathbf{z}_{k+1}-\mathcal{F}_{\theta}(\mathbf{z}_k,\mu)
    \label{eq:defect}
\end{equation}
is precisely the one-step model error analyzed in \cref{subsec:theory_decoupling}. The following decomposition, the standard first step of classical EKF stability analyses~\cite{song1995extended,boutayeb1997convergence,reif1999stochastic}, makes the role of the two curvature residuals explicit.

\begin{lemma}[Local EKF error decomposition]
\label{lem:local_ekf_error_decomposition}
Let $\{\mathbf{z}_k\}$ be any reference sequence with one-step defect $\boldsymbol{\delta}_k$ as in \cref{eq:defect}, and define the observation disturbance as the residual
\begin{equation}
    \tilde{\boldsymbol{\eta}}_{k+1} = \mathbf{y}_{k+1} - h(\mathbf{z}_{k+1}),
\end{equation}
where $f:\mathbb R^r\to\mathbb R^r$ and $h:\mathbb R^r\to\mathbb R^m$ are twice continuously differentiable on a convex neighborhood $\Omega$ containing the line segments used below. Neither relation is an assumption: $\mathbf{z}_{k+1}=f(\mathbf{z}_k)+\boldsymbol{\delta}_k$ and $\mathbf{y}_{k+1}=h(\mathbf{z}_{k+1})+\tilde{\boldsymbol{\eta}}_{k+1}$ hold by definition, whether or not the encoded latent state possesses autonomous dynamics of its own and whether or not the decoded reconstruction matches the true field at the sensor locations; the composition of $\tilde{\boldsymbol{\eta}}$ is made explicit in \cref{rem:dks_instantiation}. Assume
\begin{equation}
    \sup_{\mathbf{z}\in\Omega}\Big(\sum_{i=1}^{r}\big\|\nabla^{2}f_{i}(\mathbf{z})\big\|^{2}\Big)^{1/2}\le M_f,
    \qquad
    \sup_{\mathbf{z}\in\Omega}\Big(\sum_{i=1}^{m}\big\|\nabla^{2}h_{i}(\mathbf{z})\big\|^{2}\Big)^{1/2}\le M_h.
\end{equation}
The EKF forecast and analysis steps are
\begin{align}
    \hat{\mathbf{z}}_{k+1}^f &= f(\hat{\mathbf{z}}_k^a), \\
    \mathbf{F}_k &= Df(\hat{\mathbf{z}}_k^a), \\
    \mathbf{J}_{k+1} &= Dh(\hat{\mathbf{z}}_{k+1}^f), \\
    \hat{\mathbf{z}}_{k+1}^a
    &=
    \hat{\mathbf{z}}_{k+1}^f
    +\mathcal{K}_{k+1}\bigl(\mathbf{y}_{k+1}-h(\hat{\mathbf{z}}_{k+1}^f)\bigr),
\end{align}
where $\mathcal{K}_{k+1}$ is the EKF gain. Define
\begin{equation}
    \mathbf{e}_k^a=\mathbf{z}_k-\hat{\mathbf{z}}_k^a,
    \qquad
    \mathbf{e}_{k+1}^f=\mathbf{z}_{k+1}-\hat{\mathbf{z}}_{k+1}^f.
\end{equation}
Then there exist remainders $\boldsymbol{\rho}_k^f$ and $\boldsymbol{\rho}_{k+1}^h$ such that
\begin{align}
    \mathbf{e}_{k+1}^f
    &=\mathbf{F}_k\mathbf{e}_k^a+\boldsymbol{\delta}_k+\boldsymbol{\rho}_k^f, \\
    \mathbf{e}_{k+1}^a
    &=(\mathbf{I}-\mathcal{K}_{k+1}\mathbf{J}_{k+1})\mathbf{e}_{k+1}^f
      -\mathcal{K}_{k+1}\boldsymbol{\rho}_{k+1}^h
      -\mathcal{K}_{k+1}\tilde{\boldsymbol{\eta}}_{k+1},
\end{align}
and
\begin{align}
    \|\boldsymbol{\rho}_k^f\|
    &\le
    \frac{M_f}{2}\|\mathbf{e}_k^a\|^2, \\
    \|\boldsymbol{\rho}_{k+1}^h\|
    &\le
    \frac{M_h}{2}\|\mathbf{e}_{k+1}^f\|^2.
\end{align}
In particular, for a Koopman latent transition
\begin{equation}
    f_\mu(\mathbf{z})=\mathbf{K}(\mu;\Delta t)\mathbf{z},
\end{equation}
every component Hessian vanishes, $\nabla^{2}(f_{\mu})_{i}\equiv \mathbf{0}$, so $\mathbf{F}_k=\mathbf{K}(\mu;\Delta t)$ and
\begin{equation}
    \boldsymbol{\rho}_k^f\equiv 0.
\end{equation}
However, for the neural sensing map
\begin{equation}
    h(\mathbf{z})=\mathcal S(G_\psi(\mathbf{z})),
\end{equation}
the observation residual $\boldsymbol{\rho}_{k+1}^h$ is generally nonzero unless $\mathcal S\circ G_\psi$ is affine on the relevant local region.
\end{lemma}

For the Koopman transition, the forecast Jacobian is thus the same global matrix at every analysis state and the forecast residual vanishes identically, whereas for the MLP and xLSTM transitions the Jacobian is state dependent and the EKF sees only a local first-order approximation. This mechanism underlies the comparison between the Koopman model and the nonlinear latent baselines inside the EKF. What it does not imply: the decoded sensor map remains nonlinear, so the observation residual $\boldsymbol{\rho}^h$ survives and the filtering problem does not become linear--Gaussian; the exactly linear--Gaussian case, in which the Kalman filter is the minimum mean-square error estimator~\cite{kalman1960new,jazwinski1970stochastic}, serves only as an idealized reference point. 

The exactness statement $\boldsymbol{\rho}^f\equiv 0$ concerns the EKF's internal linearization of its own forecast model and makes no claim about the fidelity of the learned model to the true dynamics; that fidelity enters through the defect $\boldsymbol{\delta}_k$ and the observation disturbance $\tilde{\boldsymbol{\eta}}_{k+1}$, both treated in \cref{subsec:theory_decoupling}.

\subsubsection{Non-amplification of the covariance recursion}

The second mechanism concerns the covariance recursion. The lemma below records classical properties of the Kalman covariance recursion~\cite{jazwinski1970stochastic,anderson1979optimal}, proved for completeness in \cref{app:idealized_filtering_lemmas}.

\begin{lemma}[Non-amplification of the propagated covariance for a non-expansive constant linear transition]
\label{lem:ekf_covariance_nonamplification}
Consider the EKF covariance prediction and update
\begin{align}
    \mathbf{P}_{k+1}^f &= \mathbf{K} \mathbf{P}_k^a \mathbf{K}^\top+\boldsymbol{\Sigma}_{\xi}, \\
    \mathcal{K}_{k+1}
    &=\mathbf{P}_{k+1}^f\mathbf{J}_{k+1}^\top
      \left(\mathbf{J}_{k+1}\mathbf{P}_{k+1}^f\mathbf{J}_{k+1}^\top+\boldsymbol{\Sigma}_{\eta}\right)^{-1}, \\
    \mathbf{P}_{k+1}^a
    &=
    \left(\mathbf{I}-\mathcal{K}_{k+1}\mathbf{J}_{k+1}\right)\mathbf{P}_{k+1}^f,
\end{align}
where $\mathbf{P}_k^a\succeq 0$, $\boldsymbol{\Sigma}_{\eta}\succ 0$ and $\boldsymbol{\Sigma}_{\xi}\succeq 0$ is the process covariance. If
\begin{equation}
    \|\mathbf{K}\|\le 1,
\end{equation}
then the covariance matrix propagated by the EKF satisfies
\begin{equation}
    \|\mathbf{P}_{k+1}^f\|\le \|\mathbf{P}_k^a\|+\|\boldsymbol{\Sigma}_{\xi}\|.
\end{equation}
Moreover, the measurement update satisfies
\begin{equation}
    0\preceq \mathbf{P}_{k+1}^a\preceq \mathbf{P}_{k+1}^f,
\end{equation}
and consequently
\begin{equation}
    \|\mathbf{P}_{k+1}^a\|\le \|\mathbf{P}_{k+1}^f\|.
\end{equation}
\end{lemma}

Two scope remarks are needed before this is applied. The prediction step above is written with the constant Koopman transition $\mathbf{K}$; for the MLP and xLSTM baselines the corresponding step is $\mathbf{P}_{k+1}^f=\mathbf{F}_k\mathbf{P}_k^a\mathbf{F}_k^{\top}{}+\boldsymbol{\Sigma}_{\xi}$ with a state-dependent $\mathbf{F}_k$, and the same conclusion holds only under the stronger hypothesis $\sup_k\|\mathbf{F}_k\|\le 1$, which none of the baseline architectures enforces.

The non-expansiveness hypothesis is not an additional assumption for the proposed model: it is built into the architecture.

\begin{corollary}[Non-expansiveness of the generator parameterization]
\label{cor:generator_nonexpansive}
For each fixed $\mu$, the generator of \cref{subsec:koopman_model} satisfies $\mathbf{A}(\mu)+\mathbf{A}(\mu)^{\top}=-2\diag\!\big(\mathbf{d}(\mathbf{s})\circ\mathbf{d}(\mathbf{s})\big)\preceq 0$, and consequently
\begin{equation}
    \|\mathbf{K}(\mu;\Delta t)\|=\big\|e^{\mathbf{A}(\mu)\Delta t}\big\|\le 1
    \qquad\text{for all }\Delta t\ge 0.
\end{equation}
\end{corollary}

Unlike \cref{lem:local_ekf_error_decomposition}, which is stated relative to a reference trajectory, \cref{lem:ekf_covariance_nonamplification} and \cref{cor:generator_nonexpansive} are exact statements about the computation the implemented filter performs: the forecast update is precisely $\mathbf{P}_{k+1}^f=\mathbf{K}\mathbf{P}_k^a\mathbf{K}^{\top}+\boldsymbol{\Sigma}_{\xi}$, and the generator parameterization enforces $\|\mathbf{K}\|\le 1$ by construction, up to the numerical error of the matrix exponential; \cref{tab:opnorm} measures $\max_k\|\mathbf{F}_k\|_2$ directly to close this gap empirically. Together, the two results bound the growth of the covariance recursion by the process covariance per step, for the $\boldsymbol{\Sigma}_{\xi}=q\mathbf{I}$ setting used here.

Two qualifications apply. First, $\mathbf{P}_k$ is an internal quantity of the algorithm; it coincides with the covariance of the estimation error only when the model is linear and correctly specified at second order, Gaussianity being needed only for conditional-mean optimality, so non-amplification of $\mathbf{P}_k$ is a statement about the filter's uncertainty model rather than about the second moment of the realized error. Second, a strict contraction $\|\mathbf{K}\|\le a<1$, which occurs whenever all entries of the
learned dissipation $\mathbf{d}(\mathbf{s})$ are nonzero, does not drive the covariance to zero at the
process covariance actually used. Iterating the bound of
\cref{lem:ekf_covariance_nonamplification} with $\boldsymbol{\Sigma}_{\xi}=q\mathbf{I}$ gives
\begin{equation}
    \|\mathbf{P}_k^a\|\le a^{2k}\|\mathbf{P}_0^a\|+q\,\frac{1-a^{2k}}{1-a^{2}},
\end{equation}
a uniform \emph{upper} bound on the analysis covariance, with limit $q/(1-a^{2})$. Separately, the
positive process covariance gives the forecast lower bound
$\mathbf{P}_{k+1}^{f}=\mathbf{K}\mathbf{P}_k^a\mathbf{K}^{\top}+q\mathbf{I}\succeq q\mathbf{I}$.
Neither statement bounds the analysis covariance below, and neither prevents the gain from vanishing:
the gain also depends on the observation Jacobian and is zero wherever $\mathbf{J}_{k+1}$ is. These results do not establish global convergence of the EKF or universal optimality of Koopman models.

\subsubsection{One-step model error}
\label{subsec:theory_decoupling}

The two results above take the reference sequence and its defect as given. Open-loop rollout and
filtered estimation propagate error through different operators: the rollout applies the learned
forecast repeatedly with no correction, so its error accumulates the defect $\boldsymbol{\delta}_k$
over the window, whereas the filter passes it through a different operator at each analysis step and
injects measurement information, so the two errors need not evolve alike; the analysis step is not
guaranteed to reduce the realized error pointwise. This is why
rollout accuracy and assimilation accuracy need not rank models the same way, and it is the
qualitative point the experiments of \cref{sec:results} test directly.

\begin{remark}[Instantiation on the encoder-induced latent trajectory]
\label{rem:dks_instantiation}
For the learned system, take the canonical reference trajectory \cref{eq:reference_latent}, the encoded true field; because the flow map \cref{eq:flow_map} is deterministic, the defect \cref{eq:defect} is the only forecast-side disturbance. Substituting the physical measurements \cref{eq:sensing} into the definition of the observation disturbance gives the exact decomposition
\begin{equation}
    \tilde{\boldsymbol{\eta}}_{k+1}
    =\boldsymbol{\varepsilon}_{k+1}^{h}+\boldsymbol{\eta}_{k+1},
    \qquad
    \boldsymbol{\varepsilon}_{k+1}^{h}
    =\mathcal S\big(\mathbf{x}_{k+1}-G_{\psi}(\mathbf{z}_{k+1})\big),
\end{equation}
where $\boldsymbol{\eta}$ is the physical measurement noise of \cref{eq:sensing} and $\boldsymbol{\varepsilon}^{h}$ is the reconstruction error of the autoencoder sampled at the sensor locations. The learned model error entering the filter is therefore $\bar\delta=\sup_k\|\boldsymbol{\delta}_k\|$, and the observation disturbance is $b_{\eta}=b_{\mathrm{ns}}+\bar\varepsilon_h$, where $b_{\mathrm{ns}}$ bounds the physical noise and $\bar\varepsilon_h=\sup_k\|\boldsymbol{\varepsilon}_k^{h}\|$: the sensor-sampled reconstruction error acts as additional measurement noise. All learned-error channels entering the filter, namely $\bar\delta$, $\bar\varepsilon_h$, and the curvature constants, are one-step quantities, and the first two are test-time analogues of the latent-prediction and reconstruction terms of the training objective in \cref{subsec:training_objective}. Two provisos attach to this instantiation. The containment assumption of \cref{lem:local_ekf_error_decomposition} must hold for the realized run: the encoder-induced trajectory together with the rollout and filter iterates must remain in a convex region $\Omega$ on which the curvature bounds of \cref{lem:local_ekf_error_decomposition} hold. The bounded-noise assumption is a deterministic idealization of the Gaussian noise used in the experiments; it holds per realization with $b_{\mathrm{ns}}$ equal to the largest realized noise magnitude over the assimilation window, or, for Gaussian noise with per-component standard deviation $\sigma_{\mathrm{ns}}$, with probability at least $1-\alpha$ under the choice $b_{\mathrm{ns}}=\sigma_{\mathrm{ns}}\big(\sqrt{m}+\sqrt{2\log(T/\alpha)}\big)$ over a window of length $T$.
\end{remark}

The preceding error identities concern the latent state; the decoded full-field error obeys
\begin{equation}
    \big\|\mathbf{x}_k-G_{\psi}(\hat{\mathbf{z}}_k)\big\|
    \le
    L_{G}\,\big\|\mathbf{z}_k-\hat{\mathbf{z}}_k\big\|
    +
    \big\|\mathbf{x}_k-G_{\psi}(\mathbf{z}_k)\big\|,
\end{equation}
where $L_{G}$ is a Lipschitz constant of the decoder on the region visited, so the latent bounds are multiplied by $L_G$ and the reconstruction error of the reference encoding adds an offset in this bound. The genuinely irreducible part of that offset is the decoder-range error $\inf_{\mathbf{z}}\|\mathbf{x}_k-G_{\psi}(\mathbf{z})\|$, which no filter operating in the latent space can remove; for linear-subspace models it is the representation error of the subspace.

\paragraph*{Empirical diagnostics.}
The results above rest on a model-dependent quantity that is measurable after training by
post-processing alone: the forecast Jacobian norm $\max_k\|\mathbf{F}_k\|_2$ that
\cref{lem:ekf_covariance_nonamplification} requires to be at most one. \Cref{tab:opnorm} reports
$\max_k\|\mathbf{F}_k\|_2$ measured on the test runs. For the baselines these sampled values are
empirical diagnostics rather than uniform bounds over the relevant state region: a Jacobian norm
sampled at the visited states does not bound the curvature constants of
\cref{lem:local_ekf_error_decomposition} over the segments joining the true and estimated states.

\begin{table}[!htbp]
\centering
\setlength{\tabcolsep}{7pt}
\renewcommand{\arraystretch}{1.15}
\begin{adjustbox}{max width=\linewidth}
\begin{tabular}{lcccc}
\toprule
$\max_k\|\mathbf{F}_k\|_2$ & Koopman & MLP & xLSTM & pDMD \\
\midrule
1D viscous Burgers & $1.000$ & $1.465$ & $5.928$ & $2.18$ \\
2D cylinder        & $1.000$ & $2.004$ & $9.096$ & $1.37$ \\
2D dambreak        & $1.000$ & $2.322$ & $2.428$ & $1.99$ \\
3D sphere          & $1.000$ & $4.439$ & $4.144$ & $1.27$ \\
\bottomrule
\end{tabular}
\end{adjustbox}
\caption{Largest spectral norm of the forecast Jacobian actually used by the filter, over the reported
test runs. For the Koopman model the analytical statement is
$\|\mathbf{K}(\mu;\Delta t)\|_2\le 1$ by \cref{cor:generator_nonexpansive}; the tabulated $1.000$ are
the observed values at this precision. The proposed architecture enforces the unit-norm bound, whereas
the evaluated baseline transitions exceed it at some sampled states, so the hypothesis of
\cref{lem:ekf_covariance_nonamplification} does not apply to them. Failure of a sufficient condition
is not proof of the converse: it does not establish that a baseline amplifies its realized covariance
or necessarily filters poorly. The pDMD entries are the norms of the fitted reduced operator
$\tilde{\mathbf{A}}$.}
\label{tab:opnorm}
\end{table}


The two structural properties established in this section, the exact linearization of the forecast step and the non-amplification of the covariance recursion, do not eliminate learned-model error, observation-map nonlinearity or covariance misspecification, and they do not by themselves guarantee a smaller reconstruction error. They follow from linearity together with the dissipative parameterization of \cref{eq:generator}; their empirical benefit is assessed by the experiments reported earlier in this section.

\section{Conclusion}\label{sec:conclusion}
This work studies latent-model selection for sparse flow sensing through the lens of data assimilation rather than forward rollout alone. Across the main experiments, a consistent conclusion emerges: while nonlinear latent models can be competitive as open-loop predictors, the parameter-conditioned Koopman model is the most effective latent-space model when the task requires repeated forecast--update cycles with sparse and noisy observations. Assimilation under the EKF improves the estimate for the two linear latent models on every benchmark and degrades it for the two nonlinear ones, which separates the models along the linearity of the transition rather than along forecast accuracy. In the EKF experiments, Koopman attains the lowest sensing error under both exact and inexact initialization on all four benchmarks. The ordering is preserved when the measurement noise level is varied, when the latent dimension is changed, and when the training objective is extended to multi-step prediction. Under the EnKF the nonlinear models are no longer systematically degraded by assimilation, which is consistent with their sensitivity to the local approximations the EKF makes, without isolating forecast linearization as the sole cause. Across the four benchmarks, the noise levels and latent dimensions tested, sparse sensors and approximate initialization, forecast accuracy and assimilation accuracy ranked the latent models differently, and under the EKF the model that assimilated best was the one whose latent transition is linear and non-expansive. Latent dynamics intended for sensing should therefore be chosen by their closed-loop assimilation performance, not by open-loop forecast accuracy alone.

\FloatBarrier
\begin{acknowledgments}
This work was supported by the U.S. Department of Energy under Award No. DE-SC0025425 and by the National Science Foundation (NSF) under Award No. 2616260. Shaowu Pan acknowledges support from the Google Research Scholar Program. This research used resources of the National Energy Research Scientific Computing Center (NERSC), a U.S. Department of Energy Office of Science User Facility supported by the Office of Science under Contract No. DE-AC02-05CH11231, through NERSC Award DDR-ERCAP0035696. Computational resources were also provided by the Alpha HPC cluster operated by the Empire AI Consortium, Inc., with support from Empire State Development from New York State, the Simons Foundation, and the Secunda Family Foundation. This work further used Purdue Anvil, Texas A\&M University FASTER, and Texas A\&M University ACES through allocation PHY240112 and MCH260003 from the Advanced Cyberinfrastructure Coordination Ecosystem: Services and Support (ACCESS) program, which is supported by U.S. National Science Foundation Grants Nos. 2138259, 2138286, 2138307, 2137603, and 2138296. Additional computing hardware support was provided by the NVIDIA Academic Grant Program. Yadi Cao acknowledges support from start-up funding from the Department of Computer Science and
  the Institute of Artificial Intelligence at the University of Central Florida.
\end{acknowledgments}
\section*{Data availability}
The 2D cylinder, 2D dambreak and 3D sphere datasets were generated with OpenFOAM from parameterized case setups and converted to HDF5; each comprises 90 training and 11 test trajectories of 100 snapshots resampled onto a uniform grid ($64\times 64$ in 2D, $64^{3}$ in 3D), with the cylinder and sphere parameterized by $\mathrm{Re}\in[59.5,990.5]$ and the dambreak by the water-phase kinematic viscosity $\nu_{w}\in[1.0\times10^{-3},9.9\times10^{-2}]\,\mathrm{m}^{2}\mathrm{s}^{-1}$ over the training set. The 1D viscous Burgers data are synthetic: 50 training and 10 test trajectories of 100 snapshots on $N_x=1024$ points with $\nu\in[0.01,1.0]$. The datasets, the OpenFOAM case files, the trained model weights, the filtering and analysis scripts, and the per-trajectory diagnostics behind \cref{tab:opnorm} are available from the authors on reasonable request. The code and data used in the study will be made public after publication.
\appendix
\section{Baseline latent-dynamics models}\label{app:baselines}
This appendix collects the full formulations and hyperparameters of the three baseline latent-dynamics models compared against the proposed Koopman model of \cref{subsec:koopman_model}: parametric DMD (\cref{app:pdmd}), the multilayer perceptron (\cref{app:mlp}), and the xLSTM (\cref{app:xlstm}). The MLP and xLSTM share the encoder--decoder backbone of \cref{subsec:representation} and differ from the Koopman model only in their latent dynamics; pDMD is fitted separately on a global POD basis and so differs in the representation as well.
\subsection{Parametric DMD}\label{app:pdmd}

Our parametric DMD (pDMD) baseline follows the monolithic formulation
implemented in \texttt{PyDMD}\cite{ichinaga2024pydmd,demo2018pydmd}. Unlike our Koopman model, it does not learn a
parameter-dependent latent evolution operator \(\mathbf{A}(\mu)\). Instead, it first
builds a global reduced basis for all training trajectories, then fits a
single DMD model to the stacked reduced trajectories associated with the
training parameters, and finally interpolates the predicted reduced
coefficients across parameter space.

Let \(\{\mu^{(i)}\}_{i=1}^{N_\mu}\) denote the training parameters, and let
\[
\mathbf{X}^{(i)}=
\begin{bmatrix}
\mathbf{x}^{(i)}_1 & \mathbf{x}^{(i)}_2 & \cdots & \mathbf{x}^{(i)}_{N_t}
\end{bmatrix}
\in \mathbb{R}^{N\times N_t}
\]
be the snapshot matrix for parameter \(\mu^{(i)}\), where \(N\) is the number of spatial degrees of freedom and \(N_t\) is the number of time steps. A
global POD basis \(\mathbf{U}_r \in \mathbb{R}^{N\times r}\) is computed from the
concatenated training snapshots, and each trajectory is projected to the
reduced space as
\begin{equation}
    \mathbf{Z}^{(i)} = \mathbf{U}_r^\top \mathbf{X}^{(i)} \in \mathbb{R}^{r\times N_t}.
\end{equation}

In the monolithic variant, the reduced trajectories are stacked into the
augmented matrix
\begin{equation}
    \mathcal{Z} =
    \begin{bmatrix}
        \mathbf{Z}^{(1)} \\
        \mathbf{Z}^{(2)} \\
        \vdots \\
        \mathbf{Z}^{(N_\mu)}
    \end{bmatrix}
    \in \mathbb{R}^{(N_\mu r)\times N_t}.
\end{equation}
A single DMD model is then fit to \(\mathcal{Z}\). Denoting by
\(\widehat{\mathcal{Z}}_k\) the DMD-predicted reduced coefficients at time
step \(k\), we write
\begin{equation}
    \widehat{\mathcal{Z}}_k =
    \begin{bmatrix}
        \widehat{\mathbf z}^{(1)}_k \\
        \widehat{\mathbf z}^{(2)}_k \\
        \vdots \\
        \widehat{\mathbf z}^{(N_\mu)}_k
    \end{bmatrix},
    \qquad
    \widehat{\mathbf z}^{(i)}_k \in \mathbb{R}^r.
\end{equation}

To evaluate the model at an unseen parameter \(\mu\), pDMD interpolates the
predicted reduced coefficients over parameter space. For each time step
\(k\), an interpolant \(\mathcal{I}_k\) is fit to the pairs
\((\mu^{(i)}, \widehat{\mathbf z}^{(i)}_k)\), and the reduced coefficients at
the new parameter are approximated as
\begin{equation}
    \widehat{\mathbf z}_k(\mu) = \mathcal{I}_k(\mu).
    \label{eq:pdmd_interp}
\end{equation}
The full state is then reconstructed through the global POD basis,
\begin{equation}
    \widehat{\mathbf x}_k(\mu) = \mathbf{U}_r \widehat{\mathbf z}_k(\mu).
\end{equation}

This baseline therefore combines linear time forecasting in a global reduced
space with interpolation across parameter space. Compared with our Koopman
model, it retains linear reduced dynamics but does not impose a
parameter-dependent generator structure or learn an explicit operator
\(\mathbf{A}(\mu)\) for each parameter value.

We state the online map explicitly, since it is not the monolithic operator introduced above.
At a test parameter $\mu$ the interpolant of \cref{eq:pdmd_interp} supplies the predicted reduced
coefficients $\{\hat{\mathbf z}_k(\mu)\}_{k=0}^{N_t-1}$, and the one-step operator advanced online is
the least-squares map that best propagates that parameter's own coefficient sequence,
\begin{equation}
    \tilde{\mathbf{A}}(\mu)
    =\arg\min_{\mathbf{A}}\ \sum_{k=0}^{N_t-2}
      \big\|\hat{\mathbf z}_{k+1}(\mu)-\mathbf{A}\,\hat{\mathbf z}_k(\mu)\big\|_2^{2}
    =\hat{\mathbf{Z}}_{1:}(\mu)\,\hat{\mathbf{Z}}_{:-1}(\mu)^{+},
    \label{eq:pdmd_online_operator}
\end{equation}
where $\hat{\mathbf{Z}}_{:-1}(\mu)$ and $\hat{\mathbf{Z}}_{1:}(\mu)$ collect the predicted
coefficients at the first and last $N_t-1$ times and $(\cdot)^{+}$ is the pseudoinverse. The filter of
\cref{subsec:sensing} then advances
\begin{equation}
    \hat{\mathbf z}_{k+1}=\tilde{\mathbf{A}}(\mu)\,\hat{\mathbf z}_k,
    \qquad
    h(\hat{\mathbf z}_k)=\mathcal S\,\mathbf{U}_r\,\hat{\mathbf z}_k ,
\end{equation}
so the test parameter enters the online phase through the forecast operator
$\mathbf{F}_k\equiv\tilde{\mathbf{A}}(\mu)$, which is constant in $k$ but depends on $\mu$, while the
reduced basis $\mathbf{U}_r$ is the global one shared across parameters. The observation map is
therefore affine in $\hat{\mathbf z}$ and the curvature constants of
\cref{subsec:idealized_filtering_perspective} vanish for this baseline, $M_f=M_h=0$. The monolithic
operator on $\mathbb{R}^{N_\mu r}$ belongs to the offline fit and does not enter this recursion. The
operator norms reported for pDMD in \cref{tab:opnorm} are the spectral norms of
$\tilde{\mathbf{A}}(\mu)$.
The open-loop pDMD results reported in \cref{sec:results} apply this same
$\tilde{\mathbf{A}}(\mu)$ recursively from the initial reduced state, exactly as the filter forecast
does; the interpolated sequence $\{\mathcal{I}_k(\mu)\}$ is not replayed directly. Forward run and
assimilation therefore differ only in whether the measurement update is applied, which is the
comparison made for every other model. The remaining fit choices are as follows: the monolithic DMD
and the global POD basis are both truncated at $r$, matching the latent dimension of the learned
models; the parameter interpolant $\mathcal{I}_k$ is a radial basis function interpolant with the
default multiquadric kernel; and the pseudoinverse in \cref{eq:pdmd_online_operator} is the standard
least-squares pseudoinverse at default tolerance, computed on an $r\times(N_t-1)$ coefficient matrix
that is well conditioned at the values of $r$ and $N_t$ used here.
\subsection{Multilayer perceptron}\label{app:mlp}

The MLP baseline uses a fully connected nonlinear transition map,
\begin{equation}
    \hat{\mathbf{z}}_{k+1}=f_{\theta}^{\mathrm{MLP}}(\hat{\mathbf{z}}_k,\mu).
\end{equation}
It is implemented as a fully connected network acting on $[\,\hat{\mathbf{z}}_k;\mathbf{s}(\mu)\,]$, the concatenation of the latent state with the learned parameter embedding of \cref{eq:re_embed}. This is the simplest nonlinear Markovian baseline and is included to test whether a more expressive local transition model yields better sensing performance once filtering is applied. The hyperparameters of this latent model are summarized in \cref{tab:mlp_latent_hparams}.

\begin{table}[!htbp]
    \centering
    \small
    \caption{Hyperparameters of the MLP latent dynamics model.}
    \label{tab:mlp_latent_hparams}
    \renewcommand{\arraystretch}{1.2}
    \setlength{\tabcolsep}{4pt}
    \begin{tabular}{p{3.2cm} p{3.1cm} p{5.0cm}}
        \toprule
        \textbf{Component} & \textbf{Architecture / output} & \textbf{Hyperparameters} \\
        \midrule
        Latent dynamics $f_{\theta}^{\mathrm{MLP}}(\mathbf{z}_t, \mu)$
        & MLP: $\mathbb{R}^{r}\times\mathbb{R}^{r}\to\mathbb{R}^{r}$, input $[\mathbf{z}_t;\mathbf{s}(\mu)]$
        & Hidden layers: $L_f = 4$; hidden width: $105$; activation: GELU. \\
        \bottomrule
    \end{tabular}
\end{table}
\subsection{xLSTM latent dynamics}\label{app:xlstm}

To include a nonlinear baseline with a different inductive bias from the MLP, we
also consider an xLSTM latent model
\cite{beck2024xlstm}.The implementation uses the reference \texttt{xLSTMBlockStack} configured with mLSTM blocks
only (no sLSTM blocks) and a context length of one. The recurrent state is therefore re-initialized at
every forecast step and discarded afterwards, so no memory is carried between successive latent steps
and the model acts as a gated one-step transition rather than a sequence model. Writing the cell input
as the concatenation of the current latent code and the learned parameter embedding of \cref{eq:re_embed},
$\mathbf{u}_k=[\,\mathbf{z}_k;\mathbf{s}(\mu)\,]$, the realized transition is a map on $\mathbb{R}^{r}$, and it is
this map, not an augmented recursion over $(\mathbf{h},\mathbf{c},\mathbf{n},\mathbf{m})$, that is
trained, rolled out and filtered. The signature recorded in \cref{tab:xlstm_latent_hparams} is the
operative one, and the baseline should be read as a gated one-step transition rather than a recurrent
model.

The implemented cell is the mLSTM variant of \cite{beck2024xlstm}, whose
matrix memory and covariance-style update differ from the sLSTM scalar-memory recursion; we therefore
do not reproduce the sLSTM equations here. The configuration is the reference
\texttt{xLSTMBlockStack} with \texttt{mlstm\_block} only and \texttt{slstm\_at}$=[\,]$, one block,
embedding dimension $32$, four heads, causal-convolution kernel size $4$ and query--key--value
projection block size $4$, at \texttt{context\_length}$=1$. Because the context length is one and the
recurrent state is re-initialized at every call, the memory contribution from any previous step is
absent and the realized transition is the one-step map
\begin{equation}
    \mathbf{z}_{k+1}=f_{\theta}^{\mathrm{xLSTM}}(\mathbf{z}_k,\mu)
    =\mathbf{W}_{\mathrm{out}}\,\mathrm{mLSTM}_{\theta}\!\big({\mathbf{W}_{\mathrm{in}}[\,\mathbf{z}_k;\mathbf{s}(\mu)\,]+\mathbf{b}_{\mathrm{in}}}\big)
    +\mathbf{b}_{\mathrm{out}},
\end{equation}
on $\mathbb{R}^{r}$, where $\mathbf{W}_{\mathrm{in}}$ and $\mathbf{b}_{\mathrm{in}}$ project the concatenated input to the embedding dimension, which is what is trained, rolled out and filtered. This is the map whose Jacobian
enters the EKF and whose norm is reported in \cref{tab:opnorm}. The mLSTM normalization involves a
thresholded maximum, so the $C^{2}$ hypothesis of
\cref{lem:local_ekf_error_decomposition} should be read as holding away from the switching surfaces
rather than globally.

This model is more expressive than the Koopman and pDMD baselines because the
update map is nonlinear and state dependent. The hyperparameters of this latent model are summarized in \cref{tab:xlstm_latent_hparams}.

\begin{table}[!htbp]
    \centering
    \small
    \caption{Hyperparameters of the xLSTM latent dynamics model. The signature $\mathbb{R}^{r}\times\mathbb{R}^{r}\to\mathbb{R}^{r}$ is the one used: the cell is evaluated over a single timestep, so no recurrent state persists between latent steps.}
    \label{tab:xlstm_latent_hparams}
    \renewcommand{\arraystretch}{1.2}
    \setlength{\tabcolsep}{4pt}
    \begin{tabular}{p{3.2cm} p{3.1cm} p{5.0cm}}
        \toprule
        \textbf{Component} & \textbf{Architecture / output} & \textbf{Hyperparameters} \\
        \midrule
        Latent dynamics $f_{\theta}^{\mathrm{xLSTM}}(\tilde{\mathbf{z}}_t, \mathbf{u}_t)$
        & xLSTM: $\mathbb{R}^{r}\times\mathbb{R}^{r}\to\mathbb{R}^{r}$
        & Embedding dim: $32$; number of heads: $H = 4$; number of xLSTM blocks: $B = 1$. \\
        \bottomrule
    \end{tabular}
\end{table}
\section{Proofs of the theoretical results}
\label{app:idealized_filtering_lemmas}

This appendix collects the proofs of the results stated in \cref{subsec:idealized_filtering_perspective}. The notation and norm conventions of \cref{subsec:idealized_filtering_perspective} are used throughout, and all results remain local and idealized in the sense discussed there.

\begin{proof}[Proof of \cref{lem:local_ekf_error_decomposition}]
Since
\begin{equation}
    \mathbf{z}_k=\hat{\mathbf{z}}_k^a+\mathbf{e}_k^a,
\end{equation}
Taylor's theorem at $\hat{\mathbf{z}}_k^a$ gives
\begin{equation}
    f(\mathbf{z}_k)
    =
    f(\hat{\mathbf{z}}_k^a)+Df(\hat{\mathbf{z}}_k^a)\mathbf{e}_k^a+\boldsymbol{\rho}_k^f.
\end{equation}
The integral form of the remainder is, componentwise,
\begin{equation}
    \big(\boldsymbol{\rho}_k^f\big)_i
    =
    \int_0^1(1-s)\,
    (\mathbf{e}_k^a)^{\top}\,\nabla^{2}f_{i}\big(\hat{\mathbf{z}}_k^a+s\,\mathbf{e}_k^a\big)\,\mathbf{e}_k^a \, ds,
    \qquad i=1,\ldots,r.
\end{equation}
Collecting components, applying the triangle inequality for integrals, and using the curvature bound on the segment $\hat{\mathbf{z}}_k^a+s\,\mathbf{e}_k^a\in\Omega$,
\begin{equation}
    \|\boldsymbol{\rho}_k^f\|
    \le
    \int_0^1(1-s)M_f\|\mathbf{e}_k^a\|^2\,ds
    =
    \frac{M_f}{2}\|\mathbf{e}_k^a\|^2.
\end{equation}
Using the defect identity $\mathbf{z}_{k+1}=f(\mathbf{z}_k)+\boldsymbol{\delta}_k$ and $\hat{\mathbf{z}}_{k+1}^f=f(\hat{\mathbf{z}}_k^a)$, we obtain
\begin{align}
    \mathbf{e}_{k+1}^f
    &=\mathbf{z}_{k+1}-\hat{\mathbf{z}}_{k+1}^f \\
    &=f(\mathbf{z}_k)+\boldsymbol{\delta}_k-f(\hat{\mathbf{z}}_k^a) \\
    &=Df(\hat{\mathbf{z}}_k^a)\mathbf{e}_k^a+\boldsymbol{\delta}_k+\boldsymbol{\rho}_k^f \\
    &=\mathbf{F}_k\mathbf{e}_k^a+\boldsymbol{\delta}_k+\boldsymbol{\rho}_k^f.
\end{align}

Next, since
\begin{equation}
    \mathbf{z}_{k+1}=\hat{\mathbf{z}}_{k+1}^f+\mathbf{e}_{k+1}^f,
\end{equation}
Taylor's theorem at $\hat{\mathbf{z}}_{k+1}^f$ yields
\begin{equation}
    h(\mathbf{z}_{k+1})
    =
    h(\hat{\mathbf{z}}_{k+1}^f)
    +Dh(\hat{\mathbf{z}}_{k+1}^f)\mathbf{e}_{k+1}^f
    +\boldsymbol{\rho}_{k+1}^h,
\end{equation}
with
\begin{equation}
    \|\boldsymbol{\rho}_{k+1}^h\|
    \le
    \frac{M_h}{2}\|\mathbf{e}_{k+1}^f\|^2.
\end{equation}
Therefore the innovation is
\begin{align}
    \mathbf{y}_{k+1}-h(\hat{\mathbf{z}}_{k+1}^f)
    &=h(\mathbf{z}_{k+1})+\tilde{\boldsymbol{\eta}}_{k+1}-h(\hat{\mathbf{z}}_{k+1}^f) \\
    &=\mathbf{J}_{k+1}\mathbf{e}_{k+1}^f+\boldsymbol{\rho}_{k+1}^h+\tilde{\boldsymbol{\eta}}_{k+1}.
\end{align}
Substituting this identity into the EKF update gives
\begin{align}
    \hat{\mathbf{z}}_{k+1}^a
    &=
    \hat{\mathbf{z}}_{k+1}^f
    +\mathcal{K}_{k+1}
    \bigl(\mathbf{J}_{k+1}\mathbf{e}_{k+1}^f+\boldsymbol{\rho}_{k+1}^h+\tilde{\boldsymbol{\eta}}_{k+1}\bigr).
\end{align}
Thus
\begin{align}
    \mathbf{e}_{k+1}^a
    &=\mathbf{z}_{k+1}-\hat{\mathbf{z}}_{k+1}^a \\
    &=\mathbf{e}_{k+1}^f
      -\mathcal{K}_{k+1}
       \bigl(\mathbf{J}_{k+1}\mathbf{e}_{k+1}^f+\boldsymbol{\rho}_{k+1}^h+\tilde{\boldsymbol{\eta}}_{k+1}\bigr) \\
    &=(\mathbf{I}-\mathcal{K}_{k+1}\mathbf{J}_{k+1})\mathbf{e}_{k+1}^f
      -\mathcal{K}_{k+1}\boldsymbol{\rho}_{k+1}^h
      -\mathcal{K}_{k+1}\tilde{\boldsymbol{\eta}}_{k+1}.
\end{align}
For $f_\mu(\mathbf{z})=\mathbf{K}(\mu;\Delta t)\mathbf{z}$, the second derivative is zero, so the forecast remainder vanishes identically and $\mathbf{F}_k=\mathbf{K}(\mu;\Delta t)$ for all $k$. The observation statement follows because a nonlinear decoder generally has nonzero second derivative. This proves the lemma.
\end{proof}

\begin{proof}[Proof of \cref{lem:ekf_covariance_nonamplification}]
Both inequalities are standard properties of the Kalman covariance recursion~\cite{jazwinski1970stochastic,anderson1979optimal}; we record the short argument for completeness. The forecast inequality follows directly from submultiplicativity of the spectral norm:
\begin{align}
    \|\mathbf{P}_{k+1}^f\|
    &=\|\mathbf{K}\mathbf{P}_k^a\mathbf{K}^\top+\boldsymbol{\Sigma}_{\xi}\| \\
    &\le \|\mathbf{K}\|^2\|\mathbf{P}_k^a\|+\|\boldsymbol{\Sigma}_{\xi}\| \\
    &\le \|\mathbf{P}_k^a\|+\|\boldsymbol{\Sigma}_{\xi}\|.
\end{align}

For the update step, define
\begin{equation}
    \mathbf{S}_{k+1}=\mathbf{J}_{k+1}\mathbf{P}_{k+1}^f\mathbf{J}_{k+1}^\top+\boldsymbol{\Sigma}_{\eta}.
\end{equation}
Since $\boldsymbol{\Sigma}_{\eta}\succ 0$ and $\mathbf{P}_{k+1}^f\succeq 0$, we have $\mathbf{S}_{k+1}\succ 0$. The standard Kalman covariance update can be written as
\begin{equation}
    \mathbf{P}_{k+1}^a
    =
    \mathbf{P}_{k+1}^f
    -\mathbf{P}_{k+1}^f\mathbf{J}_{k+1}^\top \mathbf{S}_{k+1}^{-1}\mathbf{J}_{k+1}\mathbf{P}_{k+1}^f.
\end{equation}
The subtracted term is positive semidefinite because
\begin{equation}
    \mathbf{P}_{k+1}^f\mathbf{J}_{k+1}^\top \mathbf{S}_{k+1}^{-1}\mathbf{J}_{k+1}\mathbf{P}_{k+1}^f
    =
    \left(\mathbf{S}_{k+1}^{-1/2}\mathbf{J}_{k+1}\mathbf{P}_{k+1}^f\right)^\top
    \left(\mathbf{S}_{k+1}^{-1/2}\mathbf{J}_{k+1}\mathbf{P}_{k+1}^f\right).
\end{equation}
Thus
\begin{equation}
    \mathbf{P}_{k+1}^a\preceq \mathbf{P}_{k+1}^f.
\end{equation}
The positive semidefiniteness $\mathbf{P}_{k+1}^a\succeq 0$ follows from the Schur complement of the joint covariance matrix
\begin{equation}
    \begin{bmatrix}
        \mathbf{P}_{k+1}^f & \mathbf{P}_{k+1}^f\mathbf{J}_{k+1}^\top \\
        \mathbf{J}_{k+1}\mathbf{P}_{k+1}^f & \mathbf{J}_{k+1}\mathbf{P}_{k+1}^f\mathbf{J}_{k+1}^\top+\boldsymbol{\Sigma}_{\eta}
    \end{bmatrix}\succeq 0,
\end{equation}
whose Schur complement with respect to $\mathbf{S}_{k+1}$ is precisely $\mathbf{P}_{k+1}^a$. Therefore
\begin{equation}
    0\preceq \mathbf{P}_{k+1}^a\preceq \mathbf{P}_{k+1}^f.
\end{equation}
For positive semidefinite matrices, Loewner monotonicity implies monotonicity of the spectral norm, and hence
\begin{equation}
    \|\mathbf{P}_{k+1}^a\|\le \|\mathbf{P}_{k+1}^f\|.
\end{equation}
This proves the lemma.
\end{proof}

\begin{proof}[Proof of \cref{cor:generator_nonexpansive}]
For each fixed parameter $\mu$, write
\begin{equation}
    \mathbf{A}(\mu)=\mathbf{B}(\mu)-\mathbf{B}(\mu)^\top-\mathbf{D}(\mu),
    \qquad
    \mathbf{D}(\mu)=\operatorname{diag}(d_1(\mu)^2,\ldots,d_r(\mu)^2)\succeq 0,
\end{equation}
where $\mathbf{B}(\mu)=\mathbf{Q}(e_\omega(\mu))$ is the reshaped output of the conservative map in \cref{subsec:koopman_model}. Then
\begin{equation}
    \mathbf{A}(\mu)+\mathbf{A}(\mu)^\top=-2\mathbf{D}(\mu)\preceq 0.
\end{equation}
Let $\mathbf{z}(t)$ solve the continuous-time linear system
\begin{equation}
    \dot{\mathbf{z}}(t)=\mathbf{A}(\mu)\mathbf{z}(t).
\end{equation}
Then
\begin{align}
    \frac{d}{dt}\|\mathbf{z}(t)\|^2
    &=\mathbf{z}(t)^\top\bigl(\mathbf{A}(\mu)+\mathbf{A}(\mu)^\top\bigr)\mathbf{z}(t) \\
    &=-2\mathbf{z}(t)^\top \mathbf{D}(\mu)\mathbf{z}(t) \\
    &\le 0.
\end{align}
Therefore $\|\mathbf{z}(t)\|\le \|\mathbf{z}(0)\|$ for all $t\ge 0$. Since $\mathbf{z}(t)=e^{\mathbf{A}(\mu)t}\mathbf{z}(0)$, it follows that
\begin{equation}
    \|e^{\mathbf{A}(\mu)t}\|\le 1,
    \qquad t\ge 0.
\end{equation}
This is an instance of the classical logarithmic-norm bound $\|e^{\mathbf{A}t}\|\le e^{\mu_2(\mathbf{A})t}$ with $\mu_2(\mathbf{A})=\lambda_{\max}\big((\mathbf{A}+\mathbf{A}^{\top})/2\big)\le 0$~\cite{soderlind2006logarithmic}.
\end{proof}

The corollary should not be invoked if the discrete map is learned directly without the generator structure, if numerical approximations break the contraction property, or if the parameter varies during the forecast step in a way not represented by a fixed generator.

\bibliographystyle{apsrev4-2}
\bibliography{dks}
\end{document}